%% file: main.tex
\documentclass{article}
\usepackage[utf8]{inputenc}

\usepackage[margin=1in]{geometry}

\usepackage{float,graphicx,verbatim,fullpage,hyperref,amssymb,amsmath,amsthm,enumerate,multicol, xspace,xcolor,mathtools,thmtools,thm-restate,cleveref,xspace,tikz,caption,subcaption,wasysym, enumitem}
\usepackage[margin=1in]{geometry}
\usepackage{algorithm, cite}
\usepackage[noend]{algpseudocode}
\usepackage{enumitem}
\usepackage{comment}

\usetikzlibrary{calc}

\providecommand{\nnreals}{\mathbb{R}_{\geq 0}}

\usepackage{mleftright}
\usepackage{hyperref}
    \hypersetup{colorlinks=true, linkcolor=blue, filecolor=magenta, urlcolor=blue, citecolor=red}

\usetikzlibrary{arrows.meta}
\tikzset{>={Latex[width=1.5mm,length=1.5mm]}}

\newfloat{procedure}{htbp}{loa}
\floatname{procedure}{Procedure}

\def\R{\mathbb{R}}

\def\N{\mathbb{N}}

\newcommand{\opt}{\textsf{OPT}}

\def\ep{\varepsilon}
\def\tO{\tilde{O}}

\newtheorem{theorem}{Theorem}[section]

\newtheorem{lemma}[theorem]{Lemma}
\newtheorem{claim}[theorem]{Claim}

\theoremstyle{definition}
\newtheorem{definition}[theorem]{Definition}
\newtheorem{remark}[theorem]{Remark}

\providecommand{\email}[1]{\href{mailto:#1}{\nolinkurl{#1}\xspace}}

\input{newcommands}

\title{Approximation Algorithms for Directed Weighted Spanners}

\author{
Elena Grigorescu\thanks{Purdue University.
 E-mail: \email{elena-g@purdue.edu}. Supported in part by NSF CCF-1910659, NSF CCF-1910411, and NSF CCF-2228814.
}
 \and
 Nithish Kumar\thanks{Purdue University. E-mail: \email{kumar410@purdue.edu}. Supported in part by NSF CCF-1910411, and NSF CCF-2228814.}
 \and
 Young-San Lin\thanks{Melbourne Business School. 
 E-mail: \email{y.lin@mbs.edu}.}
}

\begin{document}

\maketitle
\begin{abstract}

In the \emph{pairwise weighted spanner} problem, the input consists of a weighted directed graph on $n$ vertices, where each edge is assigned both a \emph{cost} and a \emph{length}. Furthermore, we are given $k$ terminal vertex pairs and a distance constraint for each pair. The goal is to find a minimum-cost subgraph in which the distance constraints are satisfied. A more restricted variant of this problem was shown to be $O(2^{{\log^{1-\ep} n}})$-hard to approximate under a standard complexity assumption, by Elkin and Peleg (Theory of Computing Systems, 2007). This general formulation captures many well-studied network connectivity problems, including spanners, distance preservers, and Steiner forests. 

We study the weighted spanner problem, in which the edges have positive \emph{integral} lengths of magnitudes that are \emph{polynomial} in $n$, while the costs are \emph{arbitrary} non-negative rational numbers. Our results include the following in the classical offline setting:
\begin{itemize}
    \item  An $\tO(n^{4/5 + \ep})$-approximation algorithm for the pairwise weighted spanner problem. When the edges have unit costs and lengths, the best previous algorithm gives an  $\tO(n^{3/5 + \ep})$-approximation, due to Chlamt{\'a}{\v{c}}, Dinitz, Kortsarz, and  Laekhanukit (Transactions on Algorithms, 2020).

    \item An $\tO(n^{1/2+\ep})$-approximation algorithm for the weighted spanner problem when the terminal pairs consist of \emph{all} vertex pairs and the distances must be preserved \emph{exactly}. When the edges have unit costs and arbitrary positive lengths, the best previous algorithm gives an $\tO(n^{1/2})$-approximation for the all-pair spanner problem, due to Berman, Bhattacharyya, Makarychev, Raskhodnikova, and Yaroslavtsev (Information and Computation, 2013).
\end{itemize}
We also prove the first results for the weighted spanners in the \emph{online} setting. In the online setting, the terminal vertex pairs arrive one at a time, in an  online fashion, and edges are required to be added irrevocably to the solution in order to satisfy the distance constraints, while approximately minimizing the cost. Our results include the following:
\begin{itemize}
    \item An $\tO(k^{1/2 + \ep})$-competitive algorithm for the online pairwise weighted spanner problem. The state-of-the-art results are an $\tilde{O}(n^{4/5})$-competitive algorithm when edges have unit costs and arbitrary positive lengths, and a $\min\{\tO(k^{1/2 + \ep}), \tO(n^{2/3 + \ep})\}$-competitive algorithm when edges have unit costs and lengths, due to Grigorescu, Lin, and Quanrud (APPROX, 2021).
    
    \item An $\tO(k^{\ep})$-competitive algorithm 
    for the online single-source (or single-sink) weighted spanner problem. Without distance constraints, this problem is equivalent to the online directed Steiner tree problem. The best previous algorithm for online directed Steiner trees is an $\tO(k^{\ep})$-competitive algorithm, due to Chakrabarty, Ene, Krishnaswamy, and Panigrahi (SICOMP, 2018).
\end{itemize}
Our online results also imply efficient approximation algorithms for the corresponding offline problems. To the best of our knowledge, these are the first approximation (online) polynomial-time algorithms with sublinear approximation (competitive) ratios for the weighted spanner problems.
\end{abstract}

\input{intro}

\input{polyspannersarxiv}

\input{allpair}

\input{online}

\input{conclusion}

\bibliographystyle{acm}
\bibliography{reference}

\appendix

\input{appendix}

\end{document}

%% file: newcommands.tex
\usepackage{environ}

\newcommand{\tuple}[1]{\left(#1\right)} 
\newcommand{\eps}{\varepsilon}

\def\*#1{\mathbf{#1}}
\def\+#1{\mathcal{#1}}

\NewEnviron{problem}[1]{
	\begin{center}\fbox{\parbox{6in}{
				{\centering\scshape #1\par}
				\parskip=1ex
				\everypar{\hangindent=1em}
				\BODY
}}\end{center}}

\newcommand{\poly}{\ensuremath{\mathsf{poly}}}
\newcommand{\polylog}{\ensuremath{\mathsf{polylog}}}

\makeatletter
\newcommand*{\inlineequation}[2][]{
  \begingroup
    \refstepcounter{equation}
    \ifx\\#1\\
    \else
      \label{#1}
    \fi
    \relpenalty=10000 
    \binoppenalty=10000 
    \ensuremath{
      #2
    }
    ~\@eqnnum
  \endgroup
}
\makeatother

\newcommand{\pwsul}{\textsc{Pairwise Weighted Spanner}\xspace}
\newcommand{\rsp}{\textsc{Restricted Shortest Path}\xspace}
\newcommand{\rcsp}{\textsc{Resource-constrained Shortest Path}\xspace}
\newcommand{\dpwj}{\textit{ Distance-preserving Weighted Junction Tree }\xspace}

\newcommand{\pwspl}{\textsc{Pairwise Weighted Spanner}\xspace}
\newcommand{\opwspl}{\textsc{Online Pairwise Weighted Spanner}\xspace}
\newcommand{\swspl}{\textsc{Single-source Weighted Spanner}\xspace}
\newcommand{\oswspl}{\textsc{Online Single-source Weighted Spanner}\xspace}

\newcommand{\awdpl}{\textsc{All-pair Weighted Distance Preserver}\xspace}

\newcommand{\slc}{\textsc{Steiner Label Cover}\xspace}
\newcommand{\oslc}{\textsc{Online Steiner Label Cover}\xspace}

%% file: intro.tex
\section{Introduction} \label{sec:intro}

We study a multi-commodity problem in directed graphs, which we call the \emph{pairwise weighted spanner} problem. In this problem, we are given a directed simple graph $G = (V,E)$ with $n$ vertices, and a set of $k$ terminal pairs $D \subseteq V \times V$. Furthermore, each edge $e \in E$ is associated with a {\em cost} given by the function $c: E \to \nnreals$ and a {\em length} given by the function  $\ell: E \to \nnreals$. We say that the graph has unit lengths if $\ell(e) = 1$ (respectively, the graph has unit costs if $c(e) = 1$) for all $e \in E$. Each pair $(s, t)\in D$ is also associated with a target distance given by a function $Dist: D \to \nnreals$. Let $H=(V(H),E(H))$ be a subgraph of $G$ and $d_H(s,t)$ denote the \emph{distance} from $s$ to $t$ in $H$, i.e., the total length of a shortest $s \leadsto t$ path of edges in $E(H)$. The cost of $H$ is $\sum_{e \in E(H)} c(e)$. The goal is to find a {\em minimum-cost subgraph} $H$ of $G$ such that the distance from $s$ to $t$ is at most $Dist(s,t)$, namely, $d_H(s,t) \le Dist(s,t)$ for each $(s,t) \in D$. 

The pairwise weighted spanner problem captures many network connectivity problems and is motivated by common scenarios, such as constructing an electricity or an internet network, which requires not only cost minimization but also a delivery time tolerance for the demands. Each edge is thus associated with two ``weights'' in this setting: the cost and the delivery time. This general formulation has been studied under many variants: when the edges have general lengths and unit costs, one may ask for sparse subgraphs that {\em exactly} maintain pairwise distances, i.e., \emph{distance preservers}, or for sparse subgraphs that {\em 
approximately} maintain pairwise distances, i.e., \emph{spanners}; when the edges have general costs and unit lengths, one may ask for cheap subgraphs that maintain pairwise connectivity, i.e., \emph{Steiner forests}.
Spanners and distance preservers are well-studied objects, which have found applicability in domains such as distributed computation \cite{Awerbuch, PelegS89}, data structures \cite{Yao1982SpacetimeTF, Alon87optimalpreprocessing}, routing schemes \cite{PelegU89a,CowenW04,RodittyTZ08,PachockiRSTW18},  approximate shorthest paths \cite{DorHZ00, Elkin05, BaswanaK10}, distance oracles \cite{BaswanaK10, Chechik15, PatrascuR14}, and property testing \cite{bhattacharyya2012transitive, AwasthiJMR16}.
Similarly, Steiner forests have been studied in the context of multicommodity network design \cite{gupta2003approximation,fleischer2006simple}, mechanism design and games \cite{konemann2008group,chawla2006optimal,roughgarden2007optimal,konemann2005primal}, computational biology \cite{pirhaji2016revealing,khurana2017genome}, and computational geometry \cite{borradaile2015polynomial,bateni2012euclidean}.

A slightly more special case of the pairwise weighted spanner problem was originally proposed by Kortsarz \cite{Kortsarz2001OnTH} and Elkin and Peleg \cite{elkin2007hardness}, where it was called the \emph{weighted $s$-spanner} problem. The precise goal in \cite{Kortsarz2001OnTH,elkin2007hardness} is to find a minimum-cost subgraph that connects \emph{all} the pairs of vertices in $G$, and each $Dist(s,t)=s \cdot d_G(s, t)$ for some integer $s$ called the {\em stretch} of the spanner. The work of \cite{elkin2007hardness} presents a comprehensive list of inapproximability results for different variants of sparse $s$-spanners. Even in the special case where edges have unit costs (i.e., the directed $s$-spanner problem defined below), the problem is hard to approximate within a factor of $O(2^{{\log^{1-\eps} n}})$ unless $NP \subseteq DTIME(n^{\operatorname{polylog} n})$. 

In the case when the edges have unit costs, the weighted $s$-spanner problem is called the \emph{directed $s$-spanner} problem. For low-stretch spanners, when $s=2$, there is a tight $\Theta(\log n)$-approximate algorithm \cite{elkin1999client,Kortsarz2001OnTH}; with unit lengths and costs, when $s=3,4$, there are $\tO(n^{1/3})$-approximation algorithms \cite{berman2013approximation,dinitz2016approximating}. For $s > 4$ with general lengths, the best known approximation is $\tO(n^{1/2})$ \cite{berman2013approximation}.

The \emph{pairwise spanner} problem considers graphs with unit edge costs, $D$ can be any subset of $V \times V$, and the target distances are general. The state-of-the-art is $\tO(n^{4/5})$-approximate for general lengths \cite{grigorescu2021online} 
and $\min\{\tO(k^{1/2+\ep}), \tO(n^{3/5+\ep})\}$-approximation for unit lengths \cite{chlamtavc2020approximating,grigorescu2021online}.

When the target distances are infinite and the edges have unit lengths, the pairwise weighted spanner problem captures the \emph{directed Steiner forest} problem. For the directed Steiner forest problem, there is an  $\min\{\tO(k^{1/2+\ep}), \tO(n^{2/3+\ep})\}$-approximate algorithm for general costs \cite{berman2013approximation,chekuri2011set} and an $\tO(n^{4/7})$-approximate algorithm  for unit costs \cite{abboud2018reachability}. 

The extreme case when there is only one terminal pair, namely the \emph{restricted shortest path} problem, is $NP$-hard, and it admits a fully polynomial-time approximation scheme (FPTAS)\footnote{An FPTAS for a minimization problem minimizes the objective within a factor of $1+\ep$ in polynomial time. The polynomial is in terms of the input size of the problem and $1/\ep$.}  \cite{hassin1992approximation,lorenz2001simple}.

\subsection{Our contributions}

\subsubsection{Pairwise weighted spanners}

To the best of our knowledge, none of the variants studied in the literature gives efficient sublinear-factor approximation algorithms for the pairwise weighted spanner problem, even in the case of unit edge length. 
Our main result for pairwise weighted spanners is stated as follows and proved in Section \ref{sec:wps}.

\begin{restatable}{definition}{defpws} \label{def:pws}
    \pwsul

    \textbf{Instance}: A directed graph $G = (V, E)$ with $n$ vertices and edge costs $c: E \to \mathbb{Q}_{\ge 0}$, edge lengths $\ell: E \to \{1,2,3,...,\poly(n)\}$, and a set $D \subseteq V \times V$ of vertex pairs and their corresponding pairwise distance bounds $Dist: D \rightarrow \mathbb{Q}_{\geq 0}$ (where $\textit{Dist}(s,t) \geq d_G(s,t)$) for every terminal pair $(s,t) \in D$.

\textbf{Objective}: Find a min-cost subgraph $H$ of $G$ such that
    $d_H(s,t) \leq Dist(s,t)$ for all $(s,t) \in D$.
\end{restatable}

\begin{restatable}{theorem}{thmplwps}\label{thm:plwps}
     For any constant $\ep > 0$, there is a polynomial-time randomized algorithm for \pwsul with approximation ratio $\tilde{O}(n^{4/5 + \ep})$, which succeeds in resolving all pairs in $D$ with high probability.\footnote{Throughout our discussion, when we say high probability we mean probability at least $1-1/n$.}
\end{restatable}
 
The \pwsul problem is equivalent to the problem of finding a minimum-cost Steiner forest under pairwise distance constraints, and hence our result is the first polynomial-time $o(n)$-approximate algorithm for the directed Steiner forests with distance constraints. This problem is hard to approximate within a factor of $O(2^{{\log^{1-\eps} n}})$ unless $NP \subseteq DTIME(n^{\operatorname{polylog} n})$ even for the special case when all vertex pairs are required to be connected and the stretch $s \ge 5$ \cite{Kortsarz2001OnTH}.

\subsubsection{All-pair weighted distance preservers}

When the target distances are the distances in the given graph $G$, the spanner problem captures the \emph{distance preserver} problem. When edges have unit costs, there exists a distance preserver of cost $O(n)$ if the number of the source vertices is $O(n^{1/3})$ \cite{bodwin2021new}. When edges have unit costs and lengths, the state-of-the-art result is $\tO(n^{3/5+\ep})$-approximate \cite{chlamtavc2020approximating}.
We consider the case where the terminal set consists of \emph{all} vertex pairs and the subgraph is required to preserve the distances of all vertex pairs. This problem is called \awdpl. Our result for this problem is stated as follows and proved in Section \ref{sec:allpair}.

\begin{restatable}{definition}{defawdp} \label{def:awdp}
    \awdpl
    
    \textbf{Instance}: A directed graph $G = (V, E)$ with edge costs $c: E \to \mathbb{Q}_{\ge 0}$, edge lengths $\ell: E \to \{1,2,3,...,\poly(n)\}$. 
    
    \textbf{Objective}: Find a min-cost subgraph $H$ of $G$ such that $d_H(s,t) = d_G(s,t)$, for all $(s,t) \in V \times V$.
\end{restatable}

\begin{restatable}{theorem}{thmawdpl}\label{thm:awdpl}
     For any constant $\ep > 0$, there is a polynomial-time randomized algorithm for \awdpl with approximation ratio $\tilde{O}(n^{1/2+\ep})$, which succeeds in resolving all pairs in $V \times V$ with high probability.
\end{restatable}

Beside distance preservers, there are other previous special-case results for the all-pair weighted spanner problem. When edges have unit costs, the state-of-the-art is an $\tO(n^{1/2})$-approximation algorithm \cite{berman2013approximation}. When there are no distance constraints, this problem is termed the \emph{minimum strongly connected subgraph} problem and is equivalent to the all-pair Steiner forest problem. This problem is $NP$-hard and does not admit a polynomial-time approximation scheme if $NP \neq P$ \cite{khuller1995approximating}. The best algorithm is a $3/2$-approximation \cite{vetta2001approximating}.

\subsubsection{Online weighted spanners}
Next, we turn to online weighted spanners. In the online problem, the directed graph, the edge lengths, and the edge costs are given offline. The vertex pairs and the corresponding target distances arrive one at a time, in an online fashion, at each time stamp. The algorithm must irrevocably select edges at each time stamp and the goal is to minimize the cost, subject to the target distance constraints. We call this problem the \opwspl\ problem. For notation convenience, the vertex pair $(s_i, t_i)$ denotes the $i$-th pair that arrives online.

\begin{restatable}{definition}{defopws} \label{def:opws}
    \opwspl

\textbf{Instance}: A directed graph $G = (V, E)$ with edge costs $c: E \to \mathbb{Q}_{\ge 0}$, edge lengths $\ell: E \to \{1,2,3,...,\poly(n)\}$, and vertex pairs $D = \{(s_i,t_i) \in V \times V \mid i \in [k]\}$ ($k$ is unknown) with their corresponding pairwise distance bounds $Dist(s_i,t_i) \in \mathbb{Q}_{\geq 0}$ (where $\textit{Dist}(s_i,t_i) \geq d_G(s_i,t_i)$) arrive online one at a time. 
    
\textbf{Objective}: Upon the arrival of $(s_i,t_i)$ with $Dist(s_i,t_i)$, construct a min-cost subgraph $H$ of $G$ such that
    $d_H(s_i,t_i) \leq Dist(s_i,t_i)$ by irrevocably adding edges from $E$.
\end{restatable}

The performance of an online algorithm is measured by its {\em competitive ratio}, namely the ratio between the cost of the online solution and that of an optimal offline solution. With unit edge costs, the best algorithm is $\tO(n^{4/5})$-competitive; with unit edge costs and lengths, the state-of-the-art is $\min\{\tO(k^{1/2 + \ep}), \tO(n^{2/3 + \ep})\}$-competitive \cite{grigorescu2021online}. Our result for \opwspl is stated as follows and proved in Section \ref{sec:online}.

\begin{restatable}{theorem}{thmklwps}\label{thm:oklwps}
     For any constant $\ep > 0$, there is a polynomial-time randomized online algorithm for \opwspl with competitive ratio $\tilde{O}(k^{1/2 + \ep})$, which succeeds in resolving all pairs in $D$ with high probability.
\end{restatable}

In a special case of \pwspl where the source vertex $s \in V$ is fixed, we call this problem \swspl. Without distance constraints, this problem is equivalent to the directed Steiner tree problem.\footnote{Throughout the paper, the term \emph{without distance constraints} means that the target distances are infinity. This is equivalent to the connectivity problem.} The best algorithm for the directed Steiner tree problem is $O(k^\ep)$-approximate \cite{charikar1999approximation}.

\begin{restatable}{definition}{defssws} \label{def:ssws}
    \swspl
   
    \textbf{Instance}: A directed graph $G = (V, E)$ with edge costs $c: E \to \mathbb{Q}_{\ge 0}$, edge lengths $\ell: E \to \{1,2,3,...,\poly(n)\}$, and a set $D \subseteq \{s\} \times V$ of vertex pairs and their corresponding pairwise distance bounds $Dist: D \rightarrow \mathbb{Q}_{\geq 0}$ (where $\textit{Dist}(s,t) \geq d_G(s,t)$) for every terminal pair $(s,t) \in D$). 
    
\textbf{Objective}: Find a min-cost subgraph $H$ of $G$ such that
    $d_H(s,t) \leq Dist(s,t)$ for all $(s,t) \in D$.
\end{restatable}

When $D \subseteq \{s\} \times V$, a single-source weighted spanner connects $s$ to the sinks. We say that $s$ is the \emph{root} of the single-source weighted spanner and the single-source weighted spanner is \emph{rooted at} $s$. The definition for a single-sink weighted spanner where the terminal pairs share the same sink is defined similarly.

The online version of \swspl is termed \oswspl. For notation convenience, the vertex pair $(s, t_i)$ denotes the $i$-th pair that arrives online.

\begin{restatable}{definition}{defossws} \label{def:ossws}
    \oswspl 
    
\textbf{Instance}: A directed graph $G = (V, E)$ with edge costs $c: E \to \mathbb{Q}_{\ge 0}$, edge lengths $\ell: E \to \{1,2,3,...,\poly(n)\}$, and vertex pairs $D = \{(s,t_i) \mid t_i \in V, i \in [k]\}$ ($k$ is unknown) with their corresponding pairwise distance bounds $Dist(s,t_i) \in \mathbb{Q}_{\geq 0}$ (where $\textit{Dist}(s,t_i) \geq d_G(s,t_i)$) arrive online one at a time. 
    
\textbf{Objective}: Upon the arrival of $(s,t_i)$ with $Dist(s,t_i)$, construct a min-cost subgraph $H$ of $G$ such that
    $d_H(s,t_i) \leq Dist(s,t_i)$ by irrevocably selecting edges from $E$.
\end{restatable}

The state-of-the-art result for online directed Steiner trees is $\tO(k^{\ep})$-competitive implied by a more general online buy-at-bulk network design framework \cite{cekp}. Our result is stated as follows and proved in Section \ref{sec:online}.

\begin{restatable}{theorem}{thmswps}\label{thm:oswps}
     For any constant $\ep > 0$, there is a polynomial-time randomized online algorithm for \oswspl with approximation ratio $\tilde{O}(k^{\ep})$, which succeeds in resolving all pairs in $D$ with high probability.
\end{restatable}

Our online framework essentially generalizes the online Steiner forest problem by allowing distance constraints when edge lengths are positive integers of magnitude $\poly(n)$. We note that the online algorithms also imply efficient algorithms for the corresponding offline problems with the same approximation ratios.

\subsubsection{Summary}

We summarize our main results for weighted spanners in Table \ref{table:sum} by listing the approximation (competitive) ratios and contrast them with the corresponding known approximation (competitive) ratios. We note that offline $\tO(k^{1/2 + \ep})$-approximate \pwspl and offline $\tO(k^{\ep})$-approximate \swspl can be obtained by our online algorithms.

\begin{table}[!htb]
\begin{center}
\def\arraystretch{1.2}
\begin{tabular}{|*3{l|}}
\hline
\textbf{Problem} & \textbf{Our Results} & \textbf{Previous Results}\\
\hline
\begin{tabular}{@{}l@{}} \textsc{Pairwise} \\ \textsc{Weighted} \\ \textsc{Spanner} \end{tabular} & \begin{tabular}{@{}l@{}} $\tO(n^{4/5 + \ep})$ (Thm \ref{thm:plwps}) \\ $\tO(k^{1/2+\ep})$ (Thm \ref{thm:oklwps})\end{tabular} &  \begin{tabular}{@{}l@{}} $\tO(n^{4/5})$ (unit edge costs) \cite{grigorescu2021online} \\ $\tO(n^{3/5+\ep})$ (unit edge costs and lengths) \cite{chlamtavc2020approximating} \\ $\min\{\tO(k^{1/2+\ep}),\tO(n^{2/3+\ep})\}$ (without distance constraints) \cite{chekuri2011set,berman2013approximation} \\ $\tO(n^{4/7+\ep})$ (unit edge costs and lengths, without distance \\
constraints) \cite{abboud2018reachability} \end{tabular} \\
\hline
\begin{tabular}{@{}l@{}} \textsc{All-pair} \\ \textsc{Weighted} \\ \textsc{Spanner}  \end{tabular} & \begin{tabular}{@{}l@{}} $\tO(n^{1/2 + \ep})$ (distance \\ preservers, Thm \ref{thm:awdpl})\end{tabular} & \begin{tabular}{@{}l@{}} $\tO(n^{1/2})$ (unit edge costs) \cite{berman2013approximation} \\ $3/2$ (without distance constraints) \cite{vetta2001approximating} \end{tabular} \\
\hline
\begin{tabular}{@{}l@{}} \textsc{Online} \\ \textsc{Pairwise} \\ \textsc{Weighted} \\ \textsc{Spanner} \end{tabular} &  $\tO(k^{1/2+\ep})$ (Thm \ref{thm:oklwps}) &  \begin{tabular}{@{}l@{}} $\tO(n^{4/5})$ (unit edge costs) \cite{grigorescu2021online} \\ $\min\{\tO(k^{1/2+\ep}),\tO(n^{2/3+\ep})\}$ (unit edge costs and lengths) \cite{grigorescu2021online} \\ $\tO(k^{1/2+\ep})$ (without distance constraints) \cite{cekp} \end{tabular} \\
\hline
\begin{tabular}{@{}l@{}} \textsc{Single-} \\ \textsc{Source} \\ \textsc{Weighted} \\ \textsc{Spanner} \end{tabular} & \begin{tabular}{@{}l@{}} $\tO(k^{\ep})$ (also holds \\ for online, Thm \ref{thm:oswps}) \end{tabular} &  \begin{tabular}{@{}l@{}} $O(k^{\ep})$ (without distance constraints) \cite{charikar1999approximation} \\ $\tO(k^{\ep})$ (online, without distance constraints) \cite{cekp} \end{tabular} \\
\hline
\end{tabular}
\caption{Summary of the approximation and competitive ratios. Here, $n$ refers to the number of vertices and $k$ refers to the number of terminal pairs. All edge lengths are positive integers in $\poly(n)$ and all edge costs are non-negative rational numbers.  We note that \pwspl without distance constraints is equivalent to the directed Steiner forest problem. The all-pair weighted spanner problem without distance constraints is equivalent to the all-pair Steiner forest problem or the minimum strongly connected subgraph problem.} \label{table:sum}
\end{center}
\end{table}

\subsection{High-level technical overview}

Most of the literature on approximation algorithms for \emph{offline} spanner problems \cite{dinitz2011directed,bhattacharyya2012transitive,berman2013approximation,feldman2012improved,chlamtavc2020approximating,grigorescu2021online} partition the
terminal pairs into two types: thin or thick. 
A pair $(s, t)\in D$ is {\em thin} if the graph $G^{s,t}$ induced by feasible $s \leadsto t$ paths has a small number of vertices, and {\em thick} otherwise.
To connect each thick terminal pair $(s,t)$, it is sufficient to sample vertices from the graph $G$ to hit $G^{s,t}$, and then add shortest-path in-and-out-arborescences rooted at the sampled vertices.
To connect each thin terminal pair $(s,t)$, one uses a flow-based linear program (LP) and then rounds the solution.

\subsubsection{Pairwise Weighted Spanners}

For this problem, the goal is to approximately minimize the total cost while maintaining the required distances between terminal pairs, so it turns out that the approach for directed Steiner forests \cite{feldman2012improved,berman2013approximation} is more amenable to this formulation. The approach for directed Steiner forests \cite{berman2013approximation,feldman2012improved} is slightly different, namely, thick pairs are connected by adding cheap paths that contain at least one sampled vertex. In the Steiner forest algorithms, there are usually three cases for the terminal pairs: 1) pairs that are thick and have low-cost connecting paths, 2) the majority of the remaining pairs have high-cost connecting paths, and 3) the majority of the remaining pairs have low-cost connecting paths.

With distance constraints, we have to modify the analysis for \emph{all} three cases. The actual implementation of the strategy requires several new ideas and it significantly departs from the analysis of \cite{feldman2012improved,berman2013approximation} in several aspects, as we describe below.

In our first case, we cannot simply add cheap paths because they might violate the distance requirement.
Instead, we add \emph{feasible} cheap paths that satisfy the distance requirements, in order to connect the terminal pairs. 
For this purpose, we use the restricted shortest path FPTAS from \cite{hassin1992approximation,lorenz2001simple} as our subroutine.

The remaining two cases are resolved by using an iterative greedy algorithm based on a \emph{density} argument. In each iteration, the greedy algorithm constructs a partial solution $E' \subseteq E$ with low density. We define the density of $E'$ to be the ratio of the total edge cost of $E'$ to the number of pairs connected by $E'$. Iteratively adding low-density partial solutions leads to a global solution of approximately minimum cost.

In the second case, \cite{feldman2012improved,berman2013approximation} use the low-density \emph{junction trees} (the union of an in-arboresence and an out-arboresence rooted at the same vertex) from \cite{chekuri2011set} in order to connect pairs with high-cost paths. 
However, the junction tree approximation in \cite{chekuri2011set} cannot handle the distance constraints in our setting. Fortunately, with slight modifications, the junction tree approximation from  \cite{chlamtavc2020approximating} can be made to handle our distance requirements.

In the third case, \cite{feldman2012improved,berman2013approximation} formulate an LP where each edge has an indicator variable, then round the LP solution, and argue that with high probability, one can obtain a low-density partial solution that connects the terminal pairs with cheap paths. 
Two challenges arise in our setting.
First, the LP formulation is different from the one in \cite{feldman2012improved,berman2013approximation}
because we have to handle both distance and cost requirements.
We resolve these constraints by using a different separation oracle from the previous literature \cite{hassin1992approximation,lorenz2001simple}, namely we use the FPTAS for the \emph{resource-constrained shortest path} problem from \cite{horvath2018multi} (see Section \ref{subsubsec:multi-cri} for more details). Secondly, in order to round the LP solution, we can no longer use the analysis in \cite{berman2013approximation}. This is because the LP rounding scheme uses a union bound that depends on the number of the minimal subset of edges whose removal disconnects the terminal pairs (i.e., anti-spanners). Since we have to handle both lengths and costs in the LP constraints, we consider all possible subsets of edges and this is sufficient to achieve the $\tO(n^{4/5 + \ep})$-approximation. 

\subsubsection{All-pair Weighted Distance Preservers}

For this problem, the solution takes advantage of the requirement that we have to \emph{exactly} preserve the \emph{all-pair} distances. It turns out that the strategy for the spanner problems \cite{dinitz2011directed,bhattacharyya2012transitive,berman2013approximation,chlamtavc2020approximating,grigorescu2021online} is more amenable. Recall that terminal pairs are partitioned into thin or thick.

To settle thick terminal pairs, most of the previous work that considers graphs with unit edge cost samples vertices from the graph $G$ to hit $G^{s,t}$, and then adds shortest-path in-and-out-arborescences rooted at the sampled vertices. Note that the cost of an in-arborescence or an out-arborescence is always $n-1$. However, with edge costs, it is not clear how the \emph{cheapest} shortest-path in-and-out arborescences can be obtained. Instead, we add cheap \emph{single-source and single-sink weighted distance preservers} rooted at the sampled vertices. This approach requires using the algorithm for \oswspl described 
in details in Section \ref{sec:online}. The key observation is that the terminal pairs of any single-source (single-sink) weighted distance preserver is a subset of all vertex pairs. This implies that any approximately optimal single-source (single-sink) weighted distance preserver must be cheap compared to the cost of the optimal all-pair weighted distance preserver.\footnote{
One might ask if the same approach for thick pairs works if we consider all-pair spanners instead of distance preservers. Unfortunately, for spanners, it is possible that a nice approximate solution connects some thick pairs by a cheap path that is feasible but not the shortest. Single-source and single-sink weighted distance preservers might be too expensive. This issue does not arise when edges have unit costs since the cost of any arborescence is always $n-1$.}

The approach that settles the thin pairs closely follows the algorithm in \cite{berman2013approximation}, which rounds a fractional solution of the LP for all-pair spanners. Different from \pwspl, we only have to handle the lengths in the LP constraints, so the analysis follows \cite{berman2013approximation} and we can get a better approximation ratio. Ultimately, the costs for settling thick and thin pairs are both at most an $\tO(n^{1/2 + \ep})$ factor of the optimal solution.

\subsubsection{Online Weighted Spanners} \label{subsubsec:online}
The main challenge for the online pairwise weighted spanner problem is that the standard greedy approach, which iteratively extracts low-density greedy solutions partially connecting terminal pairs, is no longer applicable. Another challenge is to handle the distance constraints for the terminal pairs that arrive online. Fortunately, the online spanner framework from \cite{grigorescu2021online} already adapts both the approach introduced in \cite{cekp}, which  constructs a \emph{collection of junction trees} in an online fashion, and the approach of \cite{chlamtavc2020approximating}, which judiciously handles distance constraints when edges have unit lengths. Our online results are obtained by extending the framework of \cite{grigorescu2021online} from graphs with unit edge costs to general edge costs.

\subsection{Related work}

\subsubsection{Resource-constrained Shortest Path} \label{subsubsec:multi-cri}

In the resource-constrained shortest path problem \cite{horvath2018multi} for directed networks, each edge is associated with $r$ non-negative weights. Each type of weight $i \in [r-1]$ is associated with a budget. The $r$-th weight denotes the cost of the edge. The goal is to find a minimum-cost path that connects the single source to the single sink without violating the $r-1$ budgets. The results of\cite{horvath2018multi} show that when $r$ is a constant, there exists an FPTAS that finds a path with a cost at most the same as the feasible minimum-cost path by violating each budget by a factor of $1+\ep$. When $r=2$, this problem is equivalent to the restricted shortest path problem \cite{lorenz2001simple,hassin1992approximation}, which has been used extensively in the LP formulations for spanners and directed Steiner forests \cite{dinitz2011directed,chlamtavc2020approximating,feldman2012improved,berman2013approximation,grigorescu2021online}. For our purpose, $r=3$ because the LP formulation implicitly considers whether there exists a feasible path between terminal pairs whose cost exceeds a given threshold.

\subsubsection{Undirected Bi-criteria Network Design}

A general class of undirected bi-criteria network problems was introduced by \cite{marathe1998bicriteria}. A more related problem to ours is the undirected Steiner tree problem. The goal is to connect a subset of vertices to a specified root vertex. In the bi-criteria problem, the distance from the root to a target vertex is required to be at most the given global threshold. \cite{marathe1998bicriteria} presented a bi-criteria algorithm for undirected Steiner trees that is $O(\log n)$-approximate and violates the distance constraints by a factor of $O(\log n)$. Following \cite{marathe1998bicriteria}, \cite{hajiaghayi2009approximating} extends this result to a more general buy-at-bulk bi-criteria network design problem, where the objective is $\polylog(n)$-approximate and violates the distance constraints by a factor of $\polylog(n)$.

Recently, \cite{haeupler2021tree,filtser2022hop} studied the tree embedding technique used for undirected network connectivity problems with \emph{hop} constraints. For a positively weighted graph with a global parameter $h \in [n]$, the \emph{hop distance} between vertices $u$ and $v$ is the minimum weight among the $u$-$v$-paths using at most $h$ edges. Under the assumption that the ratio between the maximum edge weight and the minimum edge weight is $\poly(n)$, the tree embedding technique allows a $\polylog(n)$-approximation by relaxing the hop distance within a $\polylog(n)$ factor for a rich class of undirected network connectivity problems.

\subsubsection{Other related directed network problems}

The more related directed network problems are variants of spanners and Steiner problems, including directed Steiner trees \cite{zelikovsky1997series,charikar1999approximation}, directed Steiner network \cite{chekuri2011set}, fault-tolerance spanners \cite{dinitz2011fault,dinitz2011directed}, and parameterized complexity analysis for directed $s$-spanners \cite{fomin2022parameterized}. For a comprehensive account of the vast literature, we refer the reader to the excellent survey for spanners \cite{ahmed2020graph}. 

There is an extensive list of other related directed network problems, including distance preservers \cite{bodwin2021new,chlamtavc2020approximating}, approximate distance preservers \cite{kogan2022having}, reachability preservers \cite{abboud2018reachability}, and buy-at-bulk network design \cite{antonakopoulos2010approximating}. One direction along this line of research is to study the extremal bounds for the optimal subgraph in terms of the input parameters, instead of comparing the costs of an approximate and an optimal solution \cite{bodwin2021new,kogan2022having,abboud2018reachability}. Another direction is to consider the online problem where terminal pairs arrive online and the goal is to irrevocably select edges so that the cost of the network is approximately minimized \cite{cekp,grigorescu2021online,alon2006general}.

\subsection{Organization}
In Section \ref{sec:wps}, we present the $\tO(n^{4/5 + \ep})$-approximation algorithm for \pwspl. In Section \ref{sec:allpair}, we present the $\tO(n^{1/2 + \ep})$-approximation algorithm for \awdpl. In Section \ref{sec:online}, we present the $\tO(k^{1/2 + \ep})$-competitive algorithm for \opwspl and the $\tO(k^{\ep})$-competitive algorithm for \oswspl.

%% file: polyspannersarxiv.tex
\section{Pairwise Weighted Spanners} \label{sec:wps}

We recall the definition of the \pwspl problem.

\defpws*

In this section, we prove Theorem \ref{thm:plwps}.

{\renewcommand\footnote[1]{}\thmplwps*}

For ease of presentation, we assume that we have a guess for the cost of the optimal solution - $\opt$ for that instance as in \cite{feldman2012improved,berman2013approximation}. Let $\tau$ denote the value of our guess. We set $\tau_0 = \min_{e \in E} \{c(e) \mid c(e) > 0\}$; then we carry out multiple iterations of our overall procedure setting $\tau$ to be equal to an element in $\{(\tau_0, 2 \cdot \tau_0, 4 \cdot \tau_0,\ldots,2^i \cdot \tau, \ldots, \sum_{e \in E}c(e))\}$ for those iterations. Finally, we take the cheapest spanner from across all these iterations. Thus, it is sufficient to give the approximation guarantee for the iteration when $\opt \leq  \tau \leq 2\cdot \opt$. We can obtain this guess in $O(\log (\sum_{e \in E}c(e)/\tau_0))$ iterations, which is polynomial in the input size.

We next define some notions that are commonly used in the spanner and Steiner forest literature. Let us fix some useful parameters $\beta = n^{3/5},$ and $L = \tau/n^{4/5}$. We say that a path $p(s,t)$ that connects a specific terminal pair $(s,t)$ is {\em feasible} if $\sum_{e \in p(s,t)} \ell(e) \leq \textit{Dist}(s,t)$. We say that $p(s,t)$ is {\em cheap} if the $\sum_{e \in p(s,t)} c(e) \leq L$. We say that a terminal pair $(s,t) \in D$ is {\em thick} if the {\em local graph} $G^{s,t} = (V^{s,t},E^{s,t})$ induced by the vertices on feasible $s \leadsto t$ paths of cost at most $L$ has at least $n/\beta$ vertices; we say it is {\em thin} otherwise. We note that the definitions of thick and thin pairs are slightly different from how they are defined in \cite{berman2013approximation,feldman2012improved} (see Remark \ref{re:thick_pair_change} for the precise reason). 
We say that a set $E' \subseteq E$ settles (or resolves) a pair $(s,t) \in D$ if the subgraph $(V,E')$ contains a feasible $s \leadsto t$ path.

\subsection{Resolving thick pairs} \label{section_thick}

Let $S = \{s \mid \exists t: (s,t) \in D\}$ and $T = \{t \mid \exists s: (s,t) \in D\}$. We first settle the thick pairs with high probability. We do this by sampling some vertices using  Algorithm \ref{alg:sampling_thick} and then adding some incoming paths and outgoing paths from the samples to the vertices in $S$ and $T$ respectively using Algorithm \ref{alg:resolve_thick}. We have to ensure that any path we build is both feasible and cheap and we do that with the help of a black box for the \textsc{restricted shortest path problem} (see Definition \ref{prob:rsp}) which we formally define later.

\begin{restatable}[]{claim}{samph}
\label{cl:sampling_thick}
Algorithm \ref{alg:sampling_thick} selects a set of samples $R$ such that with high probability any given thick pair $(s,t)$ has at least one vertex from its local graph in $R$.
\end{restatable}

\begin{proof}
The standard proof is deferred to Appendix \ref{sec:missingproofs}. 
\end{proof}

\begin{algorithm}[!htb]
\caption{Sample$(G(V,E))$} \label{alg:sampling_thick}
\begin{algorithmic}[1]
\State{$R \gets \phi$, $k \gets 3 \beta \ln n$.}

\State Sample $k$ vertices independently and uniformly at random and store them in the set $R$.

\State \Return $R$.
\end{algorithmic}
\end{algorithm}

In Algorithm \ref{alg:resolve_thick}, we call Algorithm \ref{alg:sampling_thick} to get a set of samples $R$. Then, for each $u \in R,s \in S,t \in T$, we try to add a shortest $s \leadsto u$ path and a shortest $u \leadsto v$ path each of cost at most $L$, of increasing length. 

For this purpose, we use the restricted shortest path (bi-criteria path) problem from  \cite{lorenz2001simple}.

\begin{definition} \label{prob:rsp}
\textsc{\rsp
}
  
\textbf{Instance}: A directed graph $G = (V,E)$, edge lengths $\ell: E \to \mathbb{R}_{\ge 0}$, edge costs $c: E \to \mathbb{R}_{\ge 0}$, a vertex pair $(s,t) \in V \times V$, and a threshold $T \in \mathbb{R}_{> 0}$.

\textbf{Objective}: Find a minimum cost $s \leadsto t$ path $P$ such that $\sum_{e \in P} \ell(e) \leq T$.
\end{definition}

The following lemma from \cite{lorenz2001simple,hassin1992approximation} gives an FPTAS for \rsp.

\begin{lemma} \emph{(\cite{lorenz2001simple,hassin1992approximation})} \label{le:bicreteria_lorenz}
    There exists an FPTAS for \rsp that gives a $(1+\ep,1)$ approximation, i.e., the path ensures that $\sum_{e \in P} \ell(e) \leq T$ and has a cost at most $1+\ep$ times the optimal.
\end{lemma}

Using  Lemma \ref{le:bicreteria_lorenz} as our black box, we binary search for a path of length between $\min_{e \in E} \{\ell(e)\}$ and $ n \cdot \max_{e \in E} \{\ell(e)\}$ that will give us a cheap $s \leadsto u$ path. Since the edge lengths and thus the path lengths are all integers, this is possible in $O(\log(n \cdot \max_{e \in E} \{\ell(e)\}))$ iterations which is polynomial in the input size. It is possible that we never find an $s \leadsto u$ path of cost less than $L$, in which case we just ignore this $(s,u)$ pair. We then do the same for all the $(u,t)$ pairs. See the full details in Algorithm \ref{alg:resolve_thick}.

\begin{remark} \label{re:thick_pair_change}
    The need to use Lemma \ref{le:bicreteria_lorenz}, as opposed to a simple shortest path algorithm, is the main difference we have from \cite{berman2013approximation} for resolving thick pairs. We have to ensure that any path we add is both cheap and feasible, unlike \cite{berman2013approximation}, which just needs to ensure that the path is cheap. Note that this is the reason our definition of thick pairs differs from that of \cite{berman2013approximation} which necessitates some more major modifications in subsequent parts of the algorithm as well.
\end{remark}

\begin{algorithm}[!htb]
\caption{Thick pairs resolver $(G(V,E),\{\ell(e),c(e)\}_{e\in E})$} \label{alg:resolve_thick}
\begin{algorithmic}[1]
\State{$R \gets \phi$, $G' \gets \phi$.}

\State{$R \gets \text{Sample}(G(V,E))$.}

\For{$u \in R$}
\For{$s \in S$}
\State $\text{Use RSP to binary search for the shortest $s  \leadsto u$ path of cost at most $L \cdot (1+\ep)$ and add it to } G'$ if it exists.
\EndFor
\EndFor

\For{$u \in R$}
\For{$t\in T$}
\State $\text{Use RSP to binary search for the shortest $u  \leadsto t$ path of cost at most $L \cdot (1+\ep)$ and add it to } G'$ if it exists.
\EndFor
\EndFor

\State \Return $G'$
\end{algorithmic}
\end{algorithm}

\begin{lemma} \label{le:thick_final}
    With high probability, the set of edges returned by Algorithm \ref{alg:resolve_thick} resolves all thick pairs in $D$ with a total cost $\tO(n^{4/5} \cdot \tau)$. Moreover, Algorithm \ref{alg:resolve_thick} runs in polynomial time.
\end{lemma}

\begin{proof}
    If some $u \in R$ was originally in the local graph $G^{s,t}$, then Algorithm \ref{alg:resolve_thick} would have added at least one $s \leadsto u \leadsto t$ path from $G^{s,t}$ that is feasible and has cost less than $2L (1+\ep)$. This is because if $u$ was in the local graph of $(s,t)$, then there exists an $s \leadsto u$ path of cost less than $L$ of some length $\ell \in [n \cdot \max_{e \in E}\{\ell(e)\}]$. Since we binary search over the possible values for $\ell$ and take the lowest possible one, we will find such a path with distance at most the minimum length of an $s \leadsto t$ path that is cheap. Note that we could have a smaller distance because we use a larger bound for cost for our path in comparison to the local graph. Using Claim \ref{cl:sampling_thick}, such a cheap and feasible path will exist with a high probability for all $(s,t) \in D$ that are thick for some samples $u \in R$. 
    
    Now we analyze the cost of Algorithm \ref{alg:resolve_thick}. The cost from all the edges we add in Algorithm \ref{alg:resolve_thick} would be $O(n \cdot k \cdot L)$. 
    This is because we pick $k$ samples and each of them needs to add an incoming and outgoing path of cost $L (1+\ep)$ to at most $n$ vertices. Plugging in the values for $k$ and $L$, we can see that the total cost would be $\tO(n^{4/5} \cdot \tau)$. In addition, note that Algorithm \ref{alg:resolve_thick} will run in polynomial time because our binary search only needs to search the integers in $[ \min_{e \in E}\{\ell(e)\},n \cdot \max_{e \in E}\{\ell(e)\}]$.\footnote{Note that Algorithm \ref{alg:resolve_thick} will still run in polynomial-time if we use an exhaustive search instead of a binary search since the edge lengths are in $\poly(n)$.}
\end{proof}

\subsection{Resolving thin pairs} \label{subsec:pws-thin}

Now we focus on the thin pairs after removing the settled thick pairs from the set $D$. We define the {\em density} of a set of edges $E'$ to be the ratio of the total cost of these edges to the number of pairs settled by those edges. We first describe how to efficiently construct a subset $K$ of edges with density $\tO(n^{4/5+\ep}) \tau / |D|$. Then we iteratively find edge sets with that density, remove the pairs, and repeat until we resolve all thin pairs. This gives a total cost of $\tO(n^{4/5+\ep} \cdot \tau)$.

We construct $K$ by building two other sets $K_1$ and $K_2$ and picking the smaller density of them. Let $H$ be an optimal solution with cost $\tau$. Let $C$ be the set of demand pairs for which the minimum cost of a feasible $s \leadsto t$ path in $H$ is at least $L$ (note that the local graph for these pairs would be empty).

We have two cases:
\begin{itemize}
    \item$|D|/2 \leq |C| \leq |D|$ and
    \item$ 0 \leq |C| < |D|/2$.
\end{itemize}

\subsubsection{When $|D|/2 \leq |C| \leq |D|$}

We will use the notion of {\em junction tree} as a black box for resolving this case. Informally, junction trees are trees that satisfy significant demand at low cost. They have a root node $r,$ a collection of paths into $r$, and paths out of $r$ to satisfy some of this demand. In our case, we also need these paths to be cheap and short. The following is a formal definition of a junction tree variant that fits the needs of our problem.

\begin{definition} \label{def:dpwj}
    \dpwj
    
    Let $G = (V, E)$ be a directed graph with edge lengths $\ell: E \to \mathbb{R}_{\ge 0}$, edge costs $c: E \to \mathbb{R}_{\ge 0}$, and a set $D \subseteq V \times V$ of ordered pairs and their corresponding pairwise distance bounds $\textit{ Dist}: D \rightarrow \R$ (where $\textit{Dist}(s,t) \geq d_G(s,t)$  for every terminal pair $(s,t) \in D$), and a root $r \in V$. We define {\em distance-preserving weighted junction tree}  to be a subgraph $H$ of $G$ that is a union of an in-arboresences and an out-arboresences both rooted at $r$ containing an $s \leadsto t$ path going through the root $r$
    of length at most $Dist(s,t)$, for one or more $(s,t) \in D$. 
\end{definition}
    
The {\em density} of a junction tree is defined as the ratio of the sum of costs of all edges in the junction tree to the number of pairs settled by the junction tree. 

\begin{claim} \label{cl:junction_tree_helper}
    If $|D|/2\leq |C| \leq |D|$, then there exists a distance-preserving weighted junction tree of density $ O\left( n^{4/5} \cdot \tau/|D|\right)$.
\end{claim}
\begin{proof}
Let $H$ be an optimal solution subgraph of $G$ that connects all the costly thin pairs. Take the paths in $H$ connecting the pairs in $C$. The sum of the costs of all such paths is at least $|C|L$. Now, let $\mu$ be the maximum number of these paths that any edge in $G$ belongs to. The sum of the costs of the paths is at most $\mu \cdot \tau$ and thus there must exist an edge belonging to $\mu \geq |C|L/\tau$ paths. Pick such an arbitrary edge and call it the {\em heavy-enough  edge}, and call its source as the {\em heavy-enough vertex}, denoted $h_v$. Now, consider a tree made by adding feasible paths from $s \in S$ to $h_v$ and $h_v$ to $t \in T$ that satisfies at least $\mu$ pairs. We do not add an edge if it is not in $H$. This ensures that the cost of this tree is less than $\tau$. This tree would connect at least $\mu$ pairs, and thus it would have a density at most $\tau/\mu = \tau^2/(|C| L)$. 

If $L=\tau/n^{4/5}$, then $\tau^2/(|C|L) = n^{4/5} \cdot \tau/|C|$. If $|D| > |C| > |D|/2$, we have $n^{4/5} \cdot \tau/|C| = O\left(n^{4/5} \cdot \tau/|D|\right)$ and thus we have proved the existence of a junction tree of the required density.
\end{proof}

\begin{remark}
    The analysis so far for this case is quite similar to \cite{berman2013approximation}. The counting argument to establish the existence of a {\em heavy-enough edge} is identical. However, the usage of this counting argument to get a structure of the required density is slightly different because we need to account for feasibility. The main difference however is the fact that we need to use distance-preserving weighted junction tree as opposed to regular junction trees.
\end{remark}

The following lemma is essentially from \cite{chlamtavc2020approximating} (Theorem 5.1). But the small yet important modifications that we need are not covered in \cite{chlamtavc2020approximating}. We defer the proof to Appendix \ref{sec:dpwj}.

\begin{restatable}[]{lemma}{juncdp}
\label{le:poly_junction_tree_weighted_distance} 
    For any constant $\ep > 0$, there is a polynomial-time approximation algorithm for the minimum density distance-preserving junction tree as long as the edge lengths are integral and polynomial in $n$. In other words, there is a polynomial time algorithm which, given a weighted directed $n$ vertex graph $G=(V,E)$ where each edge $e \in E$ has a cost $c(e) \in \nnreals$ and integral length $\ell(e) \in \{1,2,...,\poly(n)\}$, terminal pairs $P \subseteq V \times V$, and distance bounds $Dist:P \rightarrow \N$ (where $Dist(s,t) \geq d_G(s,t)$)  for every terminal pair $(s,t) \in P$, approximates the following problem to within an $O(n^\ep)$ factor:
    \begin{itemize}
        \item Find a non-empty set of edges $F \subseteq E$ minimizing the ratio:
        \begin{equation}
            \min_{r \in V}\frac{\sum_{e \in F}c(e)}{|\{(s,t) \in P| d_{F,r}(s,t) \leq Dist(s,t)\}|}
        \end{equation}
    \end{itemize}
    where $d_{F,r}(s,t)$  is the length of the shortest path using edges in $F$ which connects $s$ to $t$ while going through $r$ (if such a path exists).
\end{restatable}

\begin{lemma} \label{le:thin_costly_final}
    When $|D|/2 \leq |C| \leq |D|$, we can get a set of edges $K_1$ that has density at most $\tO(n^{4/5+\ep} \cdot \tau/|D|)$.
\end{lemma}
\begin{proof}
    From Claim \ref{cl:junction_tree_helper}, we can see that there exists a distance-preserving weighted junction tree of density at most $O(n^{4/5} \cdot \tau/|D|)$. Now, we can just use Lemma \ref{le:poly_junction_tree_weighted_distance} to get a distance-preserving weighted junction tree with density at most $\tO(n^{4/5+\ep} \cdot \tau/|D|)$ and store the edges returned by it in $K_1$.
\end{proof}

\subsubsection{When $0 \leq |C| < |D|/2$}

To handle this case, we build a linear program (LP) that fits our problem's requirements, solve it approximately with the help of a separation oracle, and finally round it to get a set of edges with density $\tO(n^{4/5+\ep}) \cdot \tau/|D|$. The linear program is quite similar to the one used in \cite{berman2013approximation,feldman2012improved}, but it has a subtle distinction that significantly changes the tools and proof techniques we have to use. We will be referring to \cite{feldman2012improved} quite frequently in this section because \cite{berman2013approximation} does not directly present a way to solve the LP (it relies on \cite{feldman2012improved} for this).

\paragraph{Building and solving the linear program}

We will need the following definition in order to set up a relevant LP. 
For $(s,t)\in D$, let $\Pi(s,t)$ be the set of all {\em feasible} $s \leadsto t$ paths of cost at most $L$, and let $\Pi = \cup_{(s,t) \in D} \Pi(s,t)$. Each edge $e$ has a capacity $x_e$, each path $p \in \Pi$ carries $f_p$ units of flow, and $y_{s,t}$ is the total flow through all paths from $s$ to $t$. Define a linear program as follows:

 \begin{equation}
\begin{aligned} \label{lp:thin_pair_original}
& \min & & \sum_{e \in E}{c(e) \cdot x_e} \\
& \text{subject to}
& & \sum_{(s,t) \in D} y_{s,t} \geq \frac{|D|}{2}, \\
& & &\sum_{\Pi (s,t) \ni p \ni e} f_p \leq x_e & \forall (s,t) \in D, e \in E,\\
& & &\sum_{p \in \Pi(s,t)} f_p = y_{s,t} & \forall (s,t) \in D, \\
& & & 0 \leq y_{s,t},f_p,x_e \leq 1 & \forall (s,t) \in D, p \in \Pi, e \in E.
\end{aligned}
\end{equation}

LP \eqref{lp:thin_pair_original} tries to connect at least $|D|/2$ pairs from $D$ using paths of cost at most $L$ while minimizing the total cost of the edges that are used. It is almost identical to the corresponding LP in \cite{berman2013approximation, feldman2012improved} for Steiner forests, except that we consider only {\em feasible} paths that are cheaper than $L$, while they consider all paths that are cheaper than $L$.

\begin{lemma}
    Let $\opt$ be the optimal value of an instance of \pwsul. Then,
    the optimal value of LP \eqref{lp:thin_pair_original} corresponding to that instance is at most $\opt$. In addition, a solution for LP \eqref{lp:thin_pair_original} of value at most $(1+\ep)\cdot \opt$ can be found in polynomial time.
\end{lemma}

\begin{proof}
    Let us consider the dual of LP \eqref{lp:thin_pair_original}.
\begin{subequations} \label{lp:thin_pair_dual}
\begin{align} 
& \max & & \sum_{e \in E}{x_e} + \sum_{(s,t) \in D}{y_{s,t}} - W \cdot \frac{|D|}{2}\\
& \text{subject to}
& &  \sum_{(s,t) \in D} z_{(s,t),e} + c(e) \leq x_e & \forall e \in E, \label{lp:poly_1}\\ 
& & & y_{s,t} + w_{s,t} \geq W & \forall (s,t) \in D, \label{lp:poly_2}\\
& & & w_{s,t} \leq \sum_{e\in p}z_{(s,t),e} & \forall (s,t) \in D, p \in \Pi(s,t), \label{lp:exponential_constraints} \\
& & & W,x_e,y_{s,t},z_{(s,t),e} \geq 0 & \forall (s,t) \in D, e \in E. \label{lp:poly_3}
\end{align}
\end{subequations}
  
 Our dual is slightly different from the dual in \cite{berman2013approximation,feldman2012improved}. Constraints in  \eqref{lp:poly_1}, \eqref{lp:poly_2}, and \eqref{lp:poly_3} are identical, but constraints in \eqref{lp:exponential_constraints} are slightly different because the set of paths $\Pi$ we consider are different from the set of paths considered in \cite{feldman2012improved}. 
As in \cite{feldman2012improved}, 
 we could find violating constraints for those constraints in \eqref{lp:poly_1}, \eqref{lp:poly_2}, and \eqref{lp:poly_3} in polynomial time. The only constraints that require care are the constraints in \eqref{lp:exponential_constraints}, which may be exponentially many. 

When we consider a single $(s,t)$ pair, \cite{feldman2012improved} pointed out that their variant of constraints in \eqref{lp:exponential_constraints} are equivalent to \rsp which is $NP$-hard, see \cite{hassin1992approximation,lorenz2001simple}. \cite{feldman2012improved} then uses an FPTAS\cite{hassin1992approximation,lorenz2001simple} for \rsp as an approximate separation oracle for those constraints. But we need a different separation oracle because the set of paths  $\Pi$ allowed in our LP have two restrictions (as opposed to \cite{feldman2012improved} which has only one) in addition to an objective. We now define the \rcsp problem that is presented in \cite{horvath2018multi}.

\begin{definition}
    \textsc{\rcsp ($k$-RCSP)}
    
  \textbf{Instance}: A directed graph $G = (V,E)$, with edge costs $c: E \to \mathbb{Q}_{\ge 0}$, and a pair $(s,t)$. For each edge $e \in E$, we have a vector $r_e = (r_{1,e},r_{2,e},\ldots,r_{k,e})$ of size $k$ where each $r_{i,e} \in \mathbb{Q}_{\ge 0} \: \forall i \in [k] $. 

\textbf{Objective}: Find a minimum cost $s \leadsto t$ path $P$ such that $\sum_{e \in P} r_{i,e} \leq R_i,  \: \forall i \in [k]$.
\end{definition}

\begin{claim}
    2-RCSP acts as a separation oracle for those constraints in equation \eqref{lp:exponential_constraints} that correspond to a specific $(s,t) \in D$.
\end{claim}

\begin{proof}
We can use one of the resource constraints in  $2$-RCSP for ensuring that the distance constraints for $(s,t)$ are satisfied and use the other resource constraint to ensure that $w_{s,t} > \sum_{e\in P}z_{(s,t),e}$. In other words, we use one resource to model the edge lengths and another to model the dual variable $z_{\{s,t\},e}$. We can now try to find a minimum cost $s \leadsto t$ path in this instance of $2$-RCSP where costs for $2$-RCSP are equivalent to the costs in our instance of \pwsul. If the minimum cost obtained when we meet these constraints is less than $L$, then we have a violating constraint and if not we do not have one.
\end{proof}

The \textsc{Resource-constrained shortest path problem} is $NP$-hard \cite{horvath2018multi}. So, we instead get a separation oracle for an approximate variant of LP \eqref{lp:thin_pair_dual}. Now, given resource constraints $R_1,R_2,\ldots,R_k$ for the \textsc{Resource-constrained shortest path problem}, let $\opt_{RCSP}$ be the cost of the minimum cost $s \leadsto t$ path that satisfies the resource constraints. An $(1;1+\ep,\ldots,1+\ep)$-approximation scheme finds an $s \leadsto t$ path whose cost is at most $\opt_{RCSP}$, but the resource constraints are satisfied up to a factor of $1+\ep$ for that path. 

\begin{lemma} \textbf{(RCSP-\cite{horvath2018multi})} \label{le:multicriteria_horvath}
 If $k$ is a constant then there exists a fully polynomial time $(1; 1 + \ep, \ldots ,1 + \ep)$-approximation scheme for the $k$-RCSP that runs in time polynomial in input size and $1/\ep$.
\end{lemma}

We have to be careful in our usage of Lemma \ref{le:multicriteria_horvath}. The FPTAS for the Restricted Shortest Path problem from \cite{lorenz2001simple} cleanly serves the requirements of \cite{feldman2012improved} as it is a $(1+\ep;1)$ approximation and it can strictly satisfy the constraints. \cite{feldman2012improved} then uses these constraints to ensure that the weight requirements are strictly met. But the FPTAS given by \cite{horvath2018multi} does not strictly satisfy the constraints since we need both the length and weight constraints to be satisfied strictly. 

We overcome this obstacle by re-purposing the objective to handle the edge weights and by carefully ensuring that any error in the path length caused by using the $1+\ep$ approximation from \cite{horvath2018multi} does not make us use an incorrect path. To ensure that we do not select an incorrect path, it is sufficient to ensure that the potential error from \cite{horvath2018multi} is less than any error that is possible in our given input graph. Since the edge lengths are positive integers, observe that for any two $s \leadsto t$ paths with different lengths, the length difference is at least one. In addition, the path lengths are at most $n \cdot \max_{e \in E} \{\ell(e)\}$. Since all edge lengths are integral and of magnitude $\poly(n)$, it is sufficient to have $\ep \leq 1/(n \cdot \max_{e \in E} \{\ell(e)\})$.
Thus, we can fix $\ep$ such that $1/\ep = O(n\cdot \poly(n))$ to ensure that the running time will remain polynomial in input size and strictly satisfy the distance constraints. 

Now, we take an approximate version of LP \eqref{lp:thin_pair_dual} which is the following LP
\begin{equation}
\begin{aligned} \label{lp:thin_pair_approximate_dual}
& \max & & \sum_{e \in E}{x_e} + \sum_{s,t}{y_{s,t}} - W \cdot \frac{|D|}{2}\\
& \text{subject to}
& &  \sum_{s,t} z_{(s,t),e} + c(e) \leq x_e & \forall e \in E, \\
& & & y_{s,t} + w_{s,t} \geq W & \forall (s,t) \in D,\\
& & & (1+\ep) \cdot w_{s,t} \leq \sum_{e\in p}z_{(s,t),e} & \forall (s,t) \in D, p \in \Pi(s,t), \\
& & & W,x_e,y_{s,t},z_{(s,t),e} \geq 0 & \forall (s,t) \in D, e \in E.
\end{aligned}
\end{equation}

We can exactly solve LP \eqref{lp:thin_pair_approximate_dual} using \cite{horvath2018multi} and thus we can also exactly solve the dual of LP \eqref{lp:thin_pair_approximate_dual} which would be:

\begin{equation}
\begin{aligned} \label{lp:thin_pair_approximate}
& \min & & \sum_{e \in E}{c(e) \cdot x_e} \\
& \text{subject to}
& & \sum_{(s,t) \in D} y_{s,t} \geq \frac{|D|}{2}, \\
& & &\sum_{\Pi (s,t) \ni P \ni e} f_p \leq x_e \cdot (1+\ep) & \forall (s,t) \in D, e \in E,\\
& & &\sum_{P \in \Pi(s,t)} f_p = y_{s,t} & \forall (s,t) \in D, \\
& & & 0 \leq y_{s,t},f_p,x_e \leq 1 & \forall (s,t) \in D, p \in \Pi, e \in E.
\end{aligned}
\end{equation}

Let $\opt(\ep)$ and $\opt$ be the optimal values to \eqref{lp:thin_pair_approximate} and \eqref{lp:thin_pair_original}, respectively. Observe that $\opt(\ep) \leq \opt$ because the constraints in LP \eqref{lp:thin_pair_approximate} are slacker than the constraints in \eqref{lp:thin_pair_original} and both these LPs are minimization LPs. Also note that if $\hat{x}(\ep)$ is a feasible solution to \eqref{lp:thin_pair_approximate}, then by replacing the value of every variable $x_e$ in $\hat{x}(\ep)$ by $\min\{1,x_e \cdot (1+\ep)\}$, we get a new solution $\hat{x}$ which is a feasible solution to \eqref{lp:thin_pair_original}. The value of the optimal solution then is at most $(1+\ep) \cdot \opt(\ep) \leq (1+\ep) \cdot \opt$.
\end{proof}

\paragraph{Rounding our solution}

Now we need to round the solution of LP \eqref{lp:thin_pair_original} appropriately to decide which edges we need to include in our final solution. The overall structure of our rounding procedure is similar to that of \cite{berman2013approximation}, but there are some important differences in the proof techniques we use here because the nature of our problem prevents us from using some of the techniques used by \cite{berman2013approximation}. Let $\{\hat{x}_e\} \cup \{\hat{y}_{s,t}\}$ be a feasible approximate solution to LP \eqref{lp:thin_pair_original}. Let $K_2$ be the set of edges obtained by running Algorithm \ref{alg:lp_rounding} on $\{\hat{x}_e\}$.

\begin{algorithm}[!htb]
\caption{Thin pair rounding [LP rounding] ($x_e$)} \label{alg:lp_rounding}
\begin{algorithmic}[1]
\State{$E'' \gets \phi$ .}

\For{$e \in E$}
\State $\text{ Add } e \text{ to } K_2 \text{ with probability} \min\{n^{4/5}\ln  n  \cdot x_e,1\};$
\EndFor

\State \Return $E''$
\end{algorithmic}
\end{algorithm}

The following lemma is an adaptation of Claim 2.3 from \cite{berman2013approximation}.  

\begin{claim} \label{le:thin_helper_1}
    Let $A \subseteq E$. If Algorithm \ref{alg:lp_rounding} receives a fractional vector $\{\hat{x}_e\}$ with nonnegative entries satisfying $\sum_{e\in A} \hat{x}_e \geq 2/5$, the probability that it outputs a set $E''$ disjoint from $A$ is at most $\exp(-2n^{4/5} \cdot \ln n /5)$.
\end{claim}

\begin{proof}
    If $A$ contains an edge $e$ which has $\hat{x}_e \geq 1/(n^{4/5}\ln n)$, then $e$ is definitely included in $E''$. 

    Otherwise, the probability that no edge in $A$ is included in $E''$ is
    \begin{equation}
        \prod_{e \in A}(1 - n^{4/5} \ln n \cdot \hat{x}_e) \leq \exp \left( -\sum_{e \in A}n^{4/5} \ln n \cdot \hat{x}_e\right) \leq \exp \left( - \frac{2}{5} n^{4/5} \ln n \right).
        \notag
    \end{equation}
\end{proof}

Let us now define anti-spanners which serve as a useful tool to analyze the rounding algorithm for our LP. Our definition of anti-spanners is slightly different from Definition 2.4 in \cite{berman2013approximation} to account for the fact we also have distance constraints.

\begin{definition} \label{def:anti}
A set $A \subseteq E$ is an anti-spanner for a terminal pair $(s,t) \in E$ if $(V, E \setminus A)$ contains no feasible path from $s$ to $t$ of cost at most $L$. If no proper subset of anti-spanner $A$ for $(s,t)$ is an anti-spanner for $(s,t)$, then $A$ is minimal. The set of all minimal anti-spanners for all thin edges is denoted by $\mathcal{A}$.
\end{definition}

The following lemma is an analogue of Claim 2.5 from \cite{berman2013approximation}.

\begin{lemma} \label{le:thin_antispanner_bound}
    Let $\mathcal{A}$ be the set of all minimal anti-spanners for thin pairs. Then $|\mathcal{A}|$ is upper-bounded by $|D|\cdot 2^{(n/\beta)^2/2}$.
\end{lemma}

\begin{proof}
    Let $PS(s,t)$ be the power set of all edges in the local graph for a given thin pair $(s,t)$. Since $(s,t)$ is a thin pair we have at most $n/\beta$ vertices and $(n/\beta)^2/2$ edges in the local graph, therefore $|PS(s,t)| \leq {2^{(n/\beta)^2/2}}$ for any $(s,t)$ that is a thin pair. Now, every anti-spanner for a specific demand pair $(s,t) \in D$ is a set of edges and therefore corresponds to an element in $PS(s,t)$. Let $PS_{\text{ thin }} = \bigcup_{(s,t)} PS(s,t)$ where $(s,t) \in D$ are thin pairs. Every anti-spanner for a thin pair is a set of edges and therefore corresponds to an element in $PS_{\text{ thin}}$. We have $|\mathcal{A}| \leq |PS_{\text{ thin }}| \leq |D| \cdot {2^{(n/\beta)^2/2}}$ which proves the lemma.
\end{proof}

\begin{remark}
    Our upper bound for the $|\mathcal{A}|$ is different from the upper bound in \cite{berman2013approximation} because we cannot use the techniques used by \cite{berman2013approximation} to give a bound for $|\mathcal{A}|$. They use some arborescences which have a one-to-one correspondence with the set of anti-spanners and it is hard to get an equivalent structure because we have both lengths and weights on each edge. This step is responsible for the fact that we do not have the $\tO(n^{2/3 + \ep})$-approximation as in \cite{berman2013approximation}.
\end{remark}

The rest of this discussion is quite similar to \cite{berman2013approximation} although the exact constants and the expressions involved are different because of the result in Lemma \ref{le:thin_antispanner_bound}. Lemma \ref{le:thin_helper_3} is similar to Lemma 5.2 from \cite{berman2013approximation}.

\begin{lemma} \label{le:thin_helper_3}
    \label{le:berman_thin_settle_rounding}
    With high probability set $K_2$ settles every thin pair $(s,t)$ with $\hat{y}_{s,t} \geq 2/5$. 
\end{lemma}

\begin{proof}
    For every thin pair $(s,t) \in D$ with $\hat{y}_{s,t} \geq 2/5$, if $A$ is an anti-spanner for $(s,t)$ then $\sum_{e \in A} \hat{x}_e \geq \sum_{P \in \Pi(s,t)} \hat{f}_p \geq 2/5$, where $\hat{f}_p$ is the value of the variable $f_p$ in LP \eqref{lp:thin_pair_original} that corresponds to the solution $\{\hat{x}_e\} \cup \{\hat{y}_{s,t}\}$. 

    By Claim \ref{le:thin_helper_1}, the probability that $A$ is disjoint from $K_2$ is at most $\exp(-2n^{4/5} \cdot \ln n /5)$. Further using Lemma \ref{le:thin_antispanner_bound}, we can bound the number of minimal anti-spanners for thin pairs and then if we apply union bound, we have the probability that $K_2$ is disjoint from any anti spanner for a thin pair is at most
    \begin{equation} \label{eq:thin_helper_eqn1}
        \exp\left(-\frac{2}{5}n^{4/5} \cdot \ln n\right) \cdot |D|\cdot 2^{(n/\beta)^2/2}.
    \end{equation}

    In the worst case, $|D|$ is $n^2$. Recall that $\beta = n^{3/5}$, we have $(n/\beta)^2 = n^{4/5}$ and thus \eqref{eq:thin_helper_eqn1} becomes

    \begin{equation} \label{eq:thin_helper_eqn2}
        \exp\left(-\frac{2}{5}\cdot n^{4/5} \cdot \ln n + \ln \left(n^2 \cdot 2 ^ {n ^ {4/5}/2}\right)\right) = \exp\left(-\Theta(n^{4/5} \ln n)\right).
        \notag
    \end{equation}
    Thus we have shown that the probability $K_2$ is disjoint from any anti-spanner for a thin pair is exponentially small when $\hat{y}_{s,t} \geq 2/5$.
\end{proof}

\begin{lemma} \label{le:thin_cheap_final}
    When $0 \leq |C| < |D|/2$, with high probability, the density of $K_2$ is at most 
    \begin{equation}
        \tO(n^{4/5} \cdot \tau/|D|).
        \notag
    \end{equation}
\end{lemma}

\begin{proof}
    Firstly notice that the expected cost of $K_2$ would be at most $n^{4/5} \ln n \cdot \tau$. We also point out that the number of pairs $(s,t) \in D $ for which $\hat{y}_{s,t} < 2/5$ is at most $5|D|/6$ because otherwise the amount of flow between all pairs is strictly less than $|D|/2$ which violates a constraint of LP \eqref{lp:thin_pair_original}. Since with high probability all pairs for which $\hat{y}_{s,t} \geq 2/5$ are satisfied, this means that the expected density of $K_2$ is at most
    \begin{equation}
        \frac{n^{4/5} \ln n \cdot \tau}{|D|/6} = \frac{6 n^{4/5} \ln n \cdot \tau}{|D|} = \frac{\tO(n^{4/5} \cdot \tau)}{|D|}.
        \notag
    \end{equation}
\end{proof}

Now we are ready to prove Theorem \ref{thm:plwps}.

\begin{proof}[Proof of Theorem \ref{thm:plwps}] 
    Using Lemma \ref{le:thick_final} we can resolve all thick pairs with high probability with cost at most $\tO(n^{4/5 + \ep})$. Then, we can make two sets of edges $K_1$ and $K_2$ using a distance-preserving weighted junction tree and by rounding the approximate solution to LP \eqref{lp:thin_pair_original}  respectively. By Lemmas \ref{le:thin_costly_final} and \ref{le:thin_cheap_final}, we can see that at least one of them will have a density at most $\tO(n^{4/5} \cdot \tau /|D|)$. If we take the cheaper among them and keep iterating we can resolve all thin pairs with a high probability and with cost at most $\tO(n^{4/5 + \ep})$.
\end{proof}

%% file: allpair.tex
\section{All-pair Weighted Distance Preservers} \label{sec:allpair}

We recall the definition of the \awdpl problem.

\defawdp*

In this section, we prove Theorem \ref{thm:awdpl}.

\thmawdpl*

Our proof structure for this subsection is very similar to that of \cite{berman2013approximation} except for our use of single sink and single source spanners.

As in Section \ref{sec:wps}, we assume that we have a guess for the cost of the optimal solution - $\opt$ for the given instance of \awdpl. Let $\tau$ denote the value of our guess. Let us set $\beta = n^{1/2}$. We say that a terminal pair $(s,t) \in D$ is {\em thick} if the {\em local graph} $G^{s,t} = (V^{s,t},E^{s,t})$ induced by the vertices on feasible paths from $s$ to $t$ has at least $n/\beta$ vertices; we say it is {\em thin} otherwise. We note that the definitions of thick and thin pairs are slightly different from how they are defined in Section \ref{sec:wps} as we only care about the feasibility of a path, not its cost. We say that a set $E' \subseteq E$ settles (or resolves) a pair $(s,t) \in D$ if the subgraph $(V,E')$ contains a feasible path from $s$ to $t$. 

\subsection{Thick pairs}
We first resolve the thick pairs by randomly sampling vertices and building single-source and single-sink spanners from the samples using Theorem \ref{thm:oswps}. We then resolve thin pairs by building a linear program, solving it, and rounding as in \cite{berman2013approximation}. As mentioned earlier, our definition of thick and thin pairs is different in this section when compared to Section \ref{sec:wps}, and this allows us to use a much simpler proof (although one that will be effective only in the case of weighted distance preservers as opposed to the more general weigthed spanners).

\begin{algorithm}[!htb]
\caption{Thick pairs resolver - Distance preserver $(G(V,E),\{\ell(e),c(e)\}_{e\in E})$} \label{alg:distance_preserver_thick}
\begin{algorithmic}[1]
\State{$R \gets \phi$, $G' \gets \phi$.}

\For{$i = 1 \text{ to } \beta \ln n$}
\State{$v \gets \text{ a uniformly random element of } V.$}
\State{$S_{v}^{source} \gets \text{ a single-source distance preserver rooted at $v$ with $D = \{v\} \times V$ }$ by using Theorem \ref{thm:oswps}.}
\State{$S_{v}^{sink} \gets \text{ a single-sink distance preserver rooted $v$ with $D = \{v\} \times V$}$ by using Theorem \ref{thm:oswps}.}
\State{$G' \gets G' \cup S_{v}^{source} \cup S_{v}^{sink}$, $R \gets R \cup \{v\}$.}
\EndFor
\State \Return $G'$
\end{algorithmic}
\end{algorithm}

\begin{lemma} \label{le:distance_preserver_thick}
    Algorithm \ref{alg:distance_preserver_thick} resolves all thick pairs for \awdpl with high probability and cost $\Tilde{O}(n^\ep \cdot \beta \cdot \opt) = \Tilde{O}( n^{1/2 + \ep} \cdot \opt)$.
\end{lemma}

\begin{proof}
        Let $\opt(S_{v}^{source})$ be the optimal costs of a single-source distance preserver rooted at $v$ with $D =\{v\} \times V$ and $\opt(S_{v}^{sink})$ be the optimal costs of a single-sink distance preserver rooted at $v$ with $D = \{v\} \times V$.
        
        We recall Theorem \ref{thm:oswps}.
        
        \thmswps*
    
        This theorem also gives an $\tO(k^\delta)$-approximation for the offline problem \swspl for any constant $\delta > 0$. Single-sink distance preservers can be obtained by simply reversing the edges. The number of terminal pairs $k = \Theta(n^2)$. By setting $\delta = \ep/2$ and the target distances to the exact distances in $G$ for all vertex pairs, we observe that the cost due to one sample in Algorithm \ref{alg:distance_preserver_thick} is at most $\tO(n^\ep (\opt(S_{v}^{source}) + \opt(S_{v}^{sink})))$. Note that a distance preserver for all pairs also serves as a distance preserver for any subset of the pairs and thus we have for any $v \in V, \, \opt(S_{v}^{sink}) \leq \opt$ and $\opt(S_{v}^{source}) \leq \opt$. Thus, using Theorem \ref{thm:oswps}, the cost of the $G'$ returned by Algorithm \ref{alg:distance_preserver_thick} is at most $|R| \cdot \Tilde{O}(n^\ep \cdot \opt) \leq \Tilde{O}(n^\ep \cdot \beta \cdot \opt)$.
        
        Using a hitting set argument very similar to Claim \ref{cl:sampling_thick}, we can see that with high probability, there is at least one sample $v$ such that there is a $s \leadsto v \leadsto t$ path for every $(s,t) \in V \times V$ where $d_{G'}(s,v)+d_{G'}(v,t)=d_G(s,t)$.

        The single-sink distance preserver gives us a $s \leadsto v$ path of length $d_G(s,v)$ and the single-source distance preserver gives us a $v \leadsto t$ path of length $d_G(v,t)$. Thus, thick pairs are resolved with high probability by the edges in $G'$.
\end{proof} 

\subsection{Thin pairs}

To resolve thin pairs, we start by redefining anti-spanners by ignoring the path costs in Definition \ref{def:anti}.

\begin{definition}
A set $A \subseteq E$ is an anti-spanner for a demand pair $(s,t) \in E$ if $(V, E \setminus A)$ contains no feasible path from $s$ to $t$. If no proper subset of anti-spanner $A$ for $(s,t)$ is an anti-spanner for $(s,t)$, then $A$ is minimal. The set of all minimal anti-spanners for all thin edges is denoted by $\mathcal{A}$.
\end{definition}
Consider the following LP which is a slightly modified version of a similar LP from \cite{berman2013approximation} (Fig. 1).

\begin{equation}
\begin{aligned} \label{lp:distsance_preserver_thin_pair}
& \min & & \sum_{e \in E}{c(e) \cdot x_e} \\
& \text{subject to}
& & \sum_{e \in A} x_e \geq 1 & \forall A \in \mathcal{A}, \\
& & & x_e \geq 0 & \forall e \in E.
\end{aligned}
\end{equation}

Let $\opt$ denote the optimal solution to the LP. We can obtain this in a way identical to \cite{berman2013approximation} as we only change the objective (which does not affect the separation oracle). Now, if $\{\hat{x_e}\}$ denotes the vector of $x_e$'s in the solution to LP \eqref{lp:distsance_preserver_thin_pair}, then add every edge $e \in E$ to $G'$ with probability $\min(\sqrt{n} \cdot \ln n \cdot \hat{x_e}, 1)$

We now state the following claim from \cite{berman2013approximation} (Corollary 2.7).

\begin{claim} \label{cl:distance_presrver_thin}
Given a feasible solution to LP \eqref{lp:distsance_preserver_thin_pair}, the rounding procedure produces a set of edges $E''$ that settles all thin pairs with high probability and has size at most $2 \opt \cdot \sqrt{n} \cdot \ln n$.
\end{claim}

\thmawdpl*
\begin{proof}[Proof of Theorem \ref{thm:awdpl}] 
    Using Lemma \ref{le:distance_preserver_thick} we can resolve all thick pairs with high probability with cost at most $\Tilde{O}( n^{1/2 + \ep} \cdot \opt)$ by running Algorithm \ref{alg:distance_preserver_thick}. Then, using Claim \ref{cl:distance_presrver_thin}, we can solve and round LP \eqref{lp:distsance_preserver_thin_pair} to resolve the thin pairs with high probability and cost $\Tilde{O}(\opt \cdot \sqrt{n})$.
\end{proof}

\begin{remark}
Note that using Algorithm \ref{alg:distance_preserver_thick} for \pwspl (even for all-pair spanners) would not work because the cost of the single-source (sink) spanners that we need to add cannot be compared to the overall optimal solution. This is because, such single-source (sink) spanners cannot have a relaxation on the distances involved, unlike the optimal solution.
\end{remark}

%% file: online.tex
\section{Online Weighted Spanners} \label{sec:online}

In the online problem, the directed graph, the edge lengths, and the edge costs are given offline. The vertex pairs and the corresponding target distances arrive one at a time, in an online fashion, at each time stamp. The algorithm must irrevocably select edges at each time stamp and the goal is to minimize the cost, subject to the target distance constraints. We recall the definition for \opwspl. For notation convenience, the vertex pair $(s_i, t_i)$ denotes the $i$-th pair that arrives online.

\defopws*

We also consider the single-source online problem.

\defossws*

This section is dedicated to proving the following theorems.

\thmklwps*

\thmswps*

The proof outline for Theorems \ref{thm:oklwps} and \ref{thm:oswps} is as follows.
\begin{enumerate}
    \item We first show that there exists an $\alpha$-approximate solution consisting of distance-preserving weighted junction trees (see Definition 
    \ref{def:dpwj}). Here, $\alpha=O(\sqrt{k})$ for \pwsul and $\alpha=1$ for \swspl.
    \item We slightly modify the online algorithm from \cite{grigorescu2021online} to find an online solution consisting of distance-preserving weighted junction trees by losing a factor of $\tO(k^\ep)$.
\end{enumerate}

The main difference between the online approach and the offline approach in Section \ref{sec:wps} is that we cannot greedily remove partial solutions to settle the terminal pairs in the online setting. Instead, we construct a distance-preserving weighted junction tree solution in an online fashion.

\begin{definition}
A \emph{distance-preserving weighted junction tree solution} is a collection of distance-preserving weighted junction trees rooted at different vertices, that satisfies all the terminal distance constraints.
\end{definition}

We construct a distance-preserving weighted junction tree solution online and compare the online objective with the optimal distance-preserving weighted junction tree solution with objective value $\opt_{junc}$. The following theorem is essentially from \cite{grigorescu2021online} for the case when the edges have unit costs and lengths. However, the slight yet important modifications that we need when edges have arbitrary positive costs and integral lengths in $\poly(n)$ are not covered in \cite{grigorescu2021online}. For completeness, we show the proof in Appendix \ref{pf:thm:alpha-comp}.

\begin{restatable}{theorem}{thmkep} \label{thm:kep}
For any constant $\ep > 0$, there exists a polynomial-time randomized online algorithm for \opwspl that constructs a distance-preserving weighted junction tree solution online with a cost at most $\tO(k^{\ep}) \opt_{junc}$ with high probability.
\end{restatable}

With this theorem, we are ready to prove Theorems \ref{thm:oklwps} and \ref{thm:oswps}.

\begin{proof}[Proof for Theorems \ref{thm:oklwps} and \ref{thm:oswps}] Let $\opt$ denote the cost of the optimal solution and $\alpha$ denote the ratio between $\opt_{junc}$ and $\opt$. It suffices to show that $\alpha=O(\sqrt{k})$ for \pwsul and $\alpha=1$ for \swspl because Theorem \ref{thm:kep} implies the existence of an $\tO(\alpha k^\ep)$-competitive online algorithm.

To show that $\alpha=1$ for \swspl, let $H$ be an optimal solution. We observe that $H$ itself is a distance-preserving weighted junction tree rooted at the source $s$ that is connected to all the $k$ sinks, so $\alpha=1$.

To show that $\alpha=O(\sqrt{k})$ for \pwsul, we use a density argument via a greedy procedure which implies an $O(\sqrt{k})$-approximate distance-preserving weighted junction tree solution. We recall the density notion in Section \ref{subsec:pws-thin}. The density of a distance-preserving weighted junction tree is its cost divided by the number of terminal pairs that it connects within the required distances.

Intuitively, we are interested in finding low-density distance-preserving weighted junction trees. We show that there always exists a distance-preserving weighted junction tree with density at most a $\sqrt{k}$ factor of the optimal density. The proof of Lemma~\ref{lem:sqrt-k-den} closely follows the one for the directed Steiner network problem in \cite{chekuri2011set} and pairwise spanners \cite{grigorescu2021online} by considering whether there is a \emph{heavy} vertex that lies in $s_i \leadsto t_i$ paths for distinct $i$ or there is a simple path with low density. The case analysis also holds when there is a distance constraint for each $(s_i,t_i)$. We provide the proof in Appendix~\ref{pf:lem:sqrt-k-den} for the sake of completeness.

\begin{restatable}{lemma}{lemsqrtkden} \label{lem:sqrt-k-den}
There exists a distance-preserving weighted junction tree $J$ with density at most $\opt / \sqrt{k}$.
\end{restatable}

Consider the procedure that finds a minimum density distance-preserving weighted junction tree in each iteration, and continues on the remaining disconnected terminal pairs. Suppose there are $t$ iterations, and after iteration $j \in [t]$, there are $n_j$ disconnected terminal pairs. Let $n_0 = k$ and $n_t = 0$. After each iteration, the minimum cost for connecting the remaining terminal pairs in the remaining graph is at most $\opt$, so the total cost of this procedure is upper-bounded by
\[
\sum_{j=1}^t \frac{(n_{j-1} - n_j)\opt}{\sqrt{n_{j-1}}} \leq \sum_{i=1}^k \frac{\opt}{\sqrt{i}} \leq \int_1^{k+1} \frac{\opt}{\sqrt{x}} dx = 2 \opt (\sqrt{k+1} - 1) = O(\sqrt{k}) \opt\]
where the first inequality uses the upper bound by considering the worst case when only one terminal pair is removed in each iteration of the procedure.
\end{proof}

\begin{remark} \label{re:online_to_offline}
The online algorithms imply efficient algorithms for the corresponding offline problems with the same approximation ratios. For offline algorithms, one can also use a greedy approach by iteratively removing the connected terminal pairs and extracting low-density distance-preserving weighted junction trees, with a poly-logarithmic factor improvement. The poly-logarithmic factor naturally appears from the challenge that the terminal pairs arrive online. This greedy approach is not applicable in the online setting.
\end{remark}

%% file: conclusion.tex
\section{Conclusion} \label{sec:conclusion}

In this paper, we presented algorithms for a variant of directed spanners that could also handle costs on edges, in addition to the more standard setting of edge lengths. The proof strategy for Theorem \ref{thm:plwps} follows a high-level structure that is similar to other results for directed Steiner forests, but involves significant obstacles in each part of the proof due to the addition of distance constraints. We overcome these obstacles by using the proper approaches. For example, the \rcsp problem from \cite{horvath2018multi} is carefully adapted for our specifics. We also needed to carefully adapt many other parts of the proof, such as the analysis of our junction-tree approximation and our rounding algorithm for the LPs, to fit the addition of distance constraints.

We also present online algorithms for \opwspl and \oswspl. We use our result for \oswspl to solve a special case of \pwspl, namely, \awdpl, and obtain a significantly better approximation for that case. 

We propose the following directions for future work:
\begin{itemize}
    \item Is it possible to get a better analysis for the rounding algorithm for Theorem \ref{thm:plwps} as in \cite{berman2013approximation}? This should improve the overall approximation factor for \pwspl in Theorem \ref{thm:plwps}. 
    \item Is there a hardness bound for \pwspl that is greater than the existing hardness bounds for Steiner forests and unit-cost spanners?
    \item Is there a better approximation factor for all-pair weighted spanners, i.e., an instance of \pwspl where $D = V \times V$?
    \item Can we get a result for pairwise weighted distance preserver that is better than using the \pwspl results in Theorems \ref{thm:plwps}?
    
\end{itemize}

%% file: appendix.tex
\section{Missing Proofs in Section \ref{sec:wps}}
\subsection{Proof of Claim \ref{cl:sampling_thick}}
\label{sec:missingproofs}

We recall Claim \ref{cl:sampling_thick} and present its proof here.

\samph*

\begin{proof}
Because $(s,t)$ is a thick pair, we have at least $n/\beta$ vertices in its local graph by definition. Let $\textsc{Event}(s,t)$ be the event that no vertex in the local graph of $(s,t)$ is sampled. We have that

     \begin{equation*}
        \mathbf{Pr}[\textsc{Event}(s,t)] = \left(1 - \frac{n/\beta}{n}\right)^{3 \beta \ln n} = \left(1 - \frac{1}{\beta}\right)^{3 \beta \ln n} \le \exp\left(-\frac{3 \beta \ln n}{\beta}\right) = \frac{1}{n^3}.
    \end{equation*}
     
     Taking the union bound over all demand pairs in $D$ where $|D| \le n^2$, the probability that there is one $(s,t) \in D$ whose local graph does not include a sampled vertex is at most $1/n$. The probability that  any $(s,t) \in D$ has at least one vertex sampled is at least $1 - 1/n$.
\end{proof}

\subsection{Weighted junction tree with distance constraints}
\label{sec:dpwj} \label{app:sec:wjt}

We recall Lemma \ref{le:poly_junction_tree_weighted_distance} and present its proof here. But we first prove a result for the same problem of the junction trees in Lemma \ref{le:poly_junction_tree_weighted_distance} where we assume all edges have unit lengths (but costs are general) and then reduce the polynomial integral length version to that problem.

\begin{lemma}
\label{le:junction_tree_weighted_distance} 
    For any constant $\delta > 0$, there is a polynomial-time approximation algorithm for the minimum density distance-preserving weighted junction tree as long as the edge lengths are uniform. In other words, there is a polynomial-time algorithm which, given a weighted directed $n$ vertex graph $G=(V,E)$ where each edge $e \in E$ has a cost $c(e) \in \nnreals$ and unit length $\ell(e) = 1$, terminal pairs $P \subseteq V \times V$, and distance bounds $Dist:P \rightarrow \N$ (where $Dist(s,t) \geq d_G(s,t)$)  for every terminal pair $(s,t) \in P$, approximates the following problem to within an $O(n^\delta)$ factor:
    \begin{itemize}
        \item Find a non-empty set of edges $F \subseteq E$ minimizing the ratio:
        \begin{equation}
            \min_{r \in V}\frac{\sum_{e \in F}c(e)}{|\{(s,t) \in P| d_{F,r}(s,t) \leq Dist(s,t)\}|}
        \end{equation}
    \end{itemize}
    where $d_{F,r}(s,t)$  is the length of the shortest path using edges in $F$ which connects $s$ to $t$ while going through $r$ (if such a path exists).\footnote{In \cite{chlamtavc2020approximating}, the numerator is $|F|$. The difference occurs because they effectively have $c(e) = 1$ for all $e \in E$.}
\end{lemma}

\begin{proof}

    In this proof, we follow the structure for the proof of Theorem 5.1 in \cite{chlamtavc2020approximating} very closely. We first define the \textsc{Minimum Density Label Cover} problem, then we reduce an instance of our problem to an instance of the \textsc{Minimum Density Label Cover} problem using a graph construction. We also state the height reduction lemma from \cite{chlamtavc2020approximating,chekuri2011set} and use it on our instance of \textsc{Minimum Density Label Cover}. Finally, we use Lemma (5.4) from \cite{chlamtavc2020approximating} to prove our result. We do this procedure with every possible root $r \in V$ and take the distance-preserving weighted junction tree with the minimum density distance-preserving weighted junction tree among all possible choices of $r$.

    We start by defining the \textsc{Minimum Density Label Cover} problem to which we will reduce our problem.

\begin{definition}
    In the \textsc{Minimum Density Label Cover} problem, we are given a directed graph $G = (V,E),$ nonnegative edge costs $w:E \to \R_{\geq 0}$, a collection of set pairs $B\subseteq 2^V \times 2^V$, and for each pair $(S,T) \in B$, a relation $R(S,T) \subseteq S \times T$. The goal is to find a set of edges $F\subseteq E$ minimizing the ratio
    \begin{equation*}
        \frac{\sum_{e \in F}w(e)}{|\{(S,T)\in B \mid \exists (s,t) \in 
 R(S,T):\text{ $F$ contains an $s \leadsto t$ path }\}|}.
    \end{equation*}
\end{definition}

We also need the height reduction lemma  from \cite{chekuri2011set,chlamtavc2020approximating}. 

\begin{restatable}{lemma}{lemhr} \label{le:height_reduction}
    (Height Reduction) Let $G = (V,E)$ be an edge-weighted directed graph with edge weights $w:E\rightarrow \R_{\geq 0}$, let $r\in V$ be a source vertex of $G$, and let $\sigma > 0 $ be some parameter. Then we can efficiently construct an edge-weighted undirected tree $\hat{T}_r$ rooted at $\hat{r}$ of height $\sigma$ and size $|V|^{O(\sigma)}$ together with edge weights $\hat{w}:E(\hat{T}_r) \to \R_{\geq 0}$, and a vertex mapping $\psi :V(\hat{T}_r) \rightarrow V(G),$ such that
    \begin{itemize}
        \item For any arboresence $J \subseteq G$ rooted at $r$, and terminal set $S \subseteq J,$ there exists a tree $\hat{J} \subseteq \hat{T}_r$ rooted at $\hat{r}$ such that letting $L(J)$ and $L(\hat{J})$ be the set of leaves of $J$ and $\hat{J}$, respectively, we have $L(J) = \psi(L(\hat{J}))$. Moreover, $w(J) \leq O(\sigma |L|^{1/\sigma} \cdot \hat{w}(\hat{J}))$.
        \item Given any tree $\hat{J} \subseteq \hat{T}_r$ rooted at $\hat{r}$, we can efficiently find an arboresence $J \subseteq G$ rooted at $r$ such that, for leaf sets $L(J)$ and $L(\hat{J})$ as above, we have $L(J) = \psi(L(\hat{J})),$ and moreover, $w(J) \leq \hat{w}(\hat{J})$.
    \end{itemize}
\end{restatable}

We do the following procedure using every possible root $r \in V$ and then take the minimum density distance-preserving weighted junction tree.

For a specific $r$, we turn our distance problem into a connectivity problem using the following reduction which closely matches \cite{chlamtavc2020approximating} except for one change. We construct a layered directed graph with vertices:

\begin{equation*}
    V_r = ((V\setminus r) \times \{-n+1,\ldots,-2,-1,1,2,\ldots,n-1\}) \cup \{(r,0)\}
\end{equation*}
and
\begin{equation*}
    E_r = \{((u,i)(v,i+1)) | (u,i),(v,i+1) \in V_r,(u,v)\in E\}.
\end{equation*}

We set the edge weights for $e\in E_r$ connecting some $((u,i)(v,i+1))$ to $w(e) = c((u,v))$. This is the only change we make from \cite{chlamtavc2020approximating} and it has very little impact on the overall proof.

For every terminal pair $(s,t) \in P$ with distance bound $Dist(s,t)$, add new vertices $(s^t,-i)$ and $(t^s,j)$ for all $i,j \geq 0$ such that $(s,-i),(t,j) \in V_r$, and for all such $i$ and $j$ add zero-weight edges $((s^t,-i)(s,-i))$ and $((t,j)(t^s,j))$. Denote this graph as $G_r$, now for every terminal pair $(s,t) \in P$ define terminal sets $S_{s,t} = \{(s^t,-i) | i \geq 0\} \cap V(G_r)$ and $T_{s,t} = \{(t^s,j) | j \geq 0\} \cap V(G_r)$ and relation $R_{s,t} = \{(s^t,-i),(t^s,j) \in S_{s,t}\times T_{s,t} | i+j \leq Dist(s,t) \}$.

In this construction, for every terminal pair $(s,t) \in P$ and label $i,$ there is a bijection between paths of length $i$ from $s$ to $r$ in $G$, and paths from $(s^t,-i)$ to $(r,0)$ in $G_r,$ and similarly a bijection between paths of length $i$ from $r$ to $t$ in $G$, and paths from $(r,0)$ to $(t^s,i)$ in $G_r$. Now, to keep track of path lengths in $G$ we can just connect the appropriate terminal pairs in $G_r$. This construction also creates disjoint terminal pairs, by creating a separate copy $s^t$ of $s$ for every terminal pair $(s,t) \in P$ that $s$ participates in (and similarly for terminals $t$).

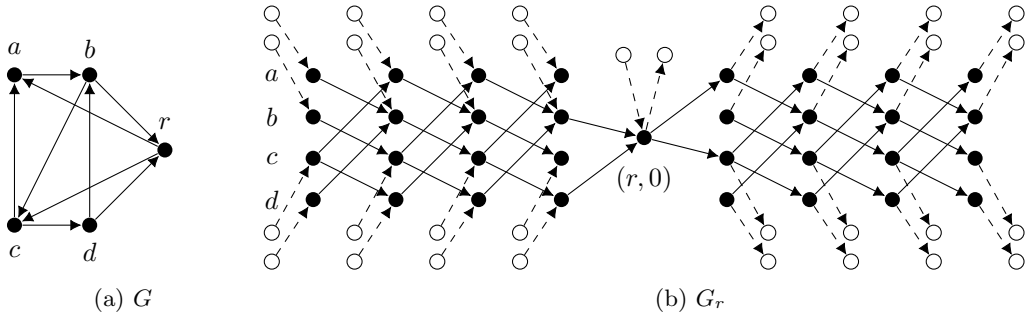
\begin{figure}[H]
\centering
\begin{subfigure}{.2\textwidth}
\begin{tikzpicture}[scale=0.5]
    \node[fill,circle, inner sep=0pt, minimum size=0.2cm] (a) at (-4,2) {};
    \node[fill,circle, inner sep=0pt, minimum size=0.2cm] (c) at (-4,-2) {};
    \node[fill,circle, inner sep=0pt, minimum size=0.2cm] (b) at (-2,2) {};
    \node[fill,circle, inner sep=0pt, minimum size=0.2cm] (d) at (-2,-2) {};
    \node[fill,circle, inner sep=0pt, minimum size=0.2cm] (r) at (0,0) {};
    \node at (-4,2.7) {$a$};
    \node at (-4,-2.7) {$c$};
    \node at (-2,2.7) {$b$};
    \node at (-2,-2.7) {$d$};
    \node at (0,0.7) {$r$};
    \path[->]
        (a) edge (b)
        (b) edge (c)
        (b) edge (r)
        (c) edge (a)
        (c) edge (d)
        (d) edge (b)
        (d) edge (r)
        (r) edge (a)
        (r) edge (c);
\end{tikzpicture}
\subcaption{$G$}
\end{subfigure}
\begin{subfigure}{.7\textwidth}
\begin{tikzpicture}[scale=0.55]
    \node at (-9,1.5) {$a$};
    \node at (-9,0.5) {$b$};
    \node at (-9,-0.5) {$c$};
    \node at (-9,-1.5) {$d$};
    \node at (0,-1) {$(r,0)$};
    \node[fill,circle, inner sep=0pt, minimum size=0.2cm] (r0) at (0,0) {};
    \foreach \x in {1,...,4}
    \foreach \y in {1,...,4}{
        \node[fill,circle, inner sep=0pt, minimum size=0.2cm] (l\x\y) at (-2*\x,-2.5 + \y) {};
        \node[fill,circle, inner sep=0pt, minimum size=0.2cm] (r\x\y) at (2*\x,-2.5 + \y) {};
    }
    \foreach \x in {1,...,4}{
        \node[draw,circle, inner sep=0pt, minimum size=0.2cm] (lo\x4) at (-2*\x-1, 3) {};
        \node[draw,circle, inner sep=0pt, minimum size=0.2cm] (lo\x3) at (-2*\x-1, 2.3) {};
        \node[draw,circle, inner sep=0pt, minimum size=0.2cm] (lo\x2) at (-2*\x-1, -2.3) {};
        \node[draw,circle, inner sep=0pt, minimum size=0.2cm] (lo\x1) at (-2*\x-1, -3) {};
        \path[->]
            (lo\x4) edge[dashed] (l\x4)
            (lo\x3) edge[dashed] (l\x3)
            (lo\x2) edge[dashed] (l\x2)
            (lo\x1) edge[dashed] (l\x1);
        \node[draw,circle, inner sep=0pt, minimum size=0.2cm] (ro\x4) at (2*\x+1, 3) {};
        \node[draw,circle, inner sep=0pt, minimum size=0.2cm] (ro\x3) at (2*\x+1, 2.3) {};
        \node[draw,circle, inner sep=0pt, minimum size=0.2cm] (ro\x2) at (2*\x+1, -2.3) {};
        \node[draw,circle, inner sep=0pt, minimum size=0.2cm] (ro\x1) at (2*\x+1, -3) {};
        \path[->]
            (r\x4) edge[dashed] (ro\x4)
            (r\x3) edge[dashed] (ro\x3)
            (r\x2) edge[dashed] (ro\x2)
            (r\x1) edge[dashed] (ro\x1);
    }
    \foreach \x in {1,...,3}{
        \pgfmathtruncatemacro{\xnext}{\x+1}
        \draw[->] (l\xnext1) -- (l\x3);
        \draw[->] (l\xnext2) -- (l\x1);
        \draw[->] (l\xnext3) -- (l\x2);
        \draw[->] (l\xnext2) -- (l\x4);
        \draw[->] (l\xnext4) -- (l\x3);
        \draw[->] (r\x1) -- (r\xnext3);
        \draw[->] (r\x2) -- (r\xnext1);
        \draw[->] (r\x3) -- (r\xnext2);
        \draw[->] (r\x2) -- (r\xnext4);
        \draw[->] (r\x4) -- (r\xnext3);
    }
    \draw[->] (l11) -- (r0);
    \draw[->] (l13) -- (r0);
    \draw[->] (r0) -- (r12);
    \draw[->] (r0) -- (r14);
    \node[draw,circle, inner sep=0pt, minimum size=0.2cm] (lr0) at (-0.5, 2) {};
    \node[draw,circle, inner sep=0pt, minimum size=0.2cm] (rr0) at (0.5, 2) {};
    \path[->]
        (lr0) edge[dashed] (r0)
        (r0) edge[dashed] (rr0);
\end{tikzpicture}
\subcaption{$G_r$}
\end{subfigure}
\caption{Construction of $G_r$ given $G=(V,E)$ and $r \in V$. The dotted lines in $G_r$ have zero weight while the solid lines in $G_r$ have the same weight as the corresponding edges in $G$. The figure is from \cite{grigorescu2021online}.}
\label{fig:lg}
\end{figure}

We finally generate a new graph $G'$ which is made up of two graphs $G_{+}$ and $G_{-}$ which are copies of $G_r$ intersecting only in the vertex $r$. For every node $u \in V$, we use $u_{+},u_{-}$ to denote the copies of $u$ in $G_{+},G_{-}$ respectively. 

The following lemma (which is also in \cite{chlamtavc2020approximating}) follows from our construction.

\begin{lemma} \label{le:junction_tree_helper}
    For any $f > 0$, and set of terminal pairs $P' \subseteq P,$ there exists an edge set $F \subseteq E(G')$ of weight $\sum_{e \in F} w(e) \leq f$ containing a path of length at most $Dist(s,t)$ from $s_{-}$ to $t_{+}$ for every $(s,t) \in P'$ iff there exists a junction tree $J \subseteq E(G_r)$ of weight $w(J) \leq f$ such that for every terminal pair $(s,t) \in P',$ $J$ contains leaves $(s^t,-i),(t^s,j)$ such that $((s^t,-i)(t^s,j)) \in R_{s,t}$. Morever, given such a junction tree $J$, we can efficiently find a corresponding edge set $F$.
\end{lemma}

Now, to prove Lemma \ref{le:junction_tree_weighted_distance}, we just need to show that we can achieve an $O(n^\ep)$ approximation for the \textsc{Minimum Density Label Cover} instance $(G_r,w,\{(S_{s,t}, T_{s,t}), R_{s,t} \mid (s,t) \in P\})$ obtained from our reduction. We apply Lemma \ref{le:height_reduction} to our weighted graph $(G_r,w)$ with constant parameter $\sigma > 1/\ep,$ and obtain a shallow tree $\hat{T}_r$ with weights $\hat{w}$, and a mapping $\psi: V(\hat{T}_r) \rightarrow V(G_r)$. All terminals $(s^t,-i),(t^s,j)$ in $G_r$ are now represented by sets of terminals $\psi^{-1}((s^t,-i))$ and $\psi^{-1}((t^s,j))$, respectively in $V(\hat{T}_r)$. We can extend the relation $R_{s,t}$ in the natural way to 
\begin{equation*}
    \hat{R}_{s,t} = \{(\hat{s},\hat{t}) \in \psi^{-1}(S_{s,t}) \times \psi^{-1}(T_{s,t}) \mid  (\psi(\hat{s}),\psi(\hat{t})) \in R_{s,t}\}.
\end{equation*}

Thus, using Lemmas \ref{le:height_reduction} and \ref{le:junction_tree_helper}, it suffices to show the following lemma which is already proved in \cite{chlamtavc2020approximating}.

\begin{lemma}
    There exists a polynomial time algorithm which, in the above setting, gives an $O(\log^3 n)$ approximation for the following problem:
    \begin{itemize}
        \item Find a tree $T \subseteq \hat{T}_r$ minimizing the ratio
        \begin{equation*}
            \frac{\hat{w}(E(T))}{|\{(s,t)\in P \mid \exists (\hat{s},\hat{t}) \in 
 \hat{R}(S,T):\text{ T contains an $\hat{s}$-$\hat{t}$ path }\}|}.
        \end{equation*}
    \end{itemize}
\end{lemma}
\end{proof}

We now adapt the above lemma to have general lengths with a very simple reduction. 

\juncdp*

\begin{proof}

Let the edge lengths in our instance be in $O(n^{\lambda})$ for some constant $\lambda > 0$. This assumption is valid because the edge lengths are in $\{1,2,...,\poly(n)\}$. Given any integer and polynomial length instance, we can reduce it to an instance with unit lengths as follows. For any edge $e = (u,v) \in E$, with length $\ell(e)$, break it into $\ell(e)$ new edges which are connected as a path that starts at $u$ and ends at $v$. Each of these new edges will have cost $c(e)/\ell(e)$ and be of unit length. We observe that solving this new instance of will give us a solution for our original problem. The whole reduction can be done in polynomial time and size because the newly created graph will have $O(n^\lambda \cdot |V|)$ vertices and $O(n^\lambda \cdot |E|)$ edges. Thus by setting $\delta = \ep/{(\lambda+1)}$ in Lemma \ref{le:junction_tree_weighted_distance}, we prove this lemma.  
\end{proof}

\section{Missing Proofs in Section \ref{sec:online}}

\subsection{Missing proof for Theorem~\ref{thm:kep}} \label{pf:thm:alpha-comp}

\thmkep*

\begin{proof}
    It suffices to show that the theorem holds even for \opwspl when all the edges have unit lengths. Suppose we have a black-box online algorithm when edges have unit lengths. We can construct an online algorithm for \opwspl by using the black box as follows. We decompose each edge $e \in E$ into a path consisting of $\ell(e)$ edges, each with unit length and cost $c(e)/\ell(e)$. The terminal vertices correspond to the endpoint vertices of the paths for the decomposed edges. The number of vertices in this graph is still polynomial in $n$ because the edge lengths are polynomial in $n$. A black-box online algorithm on graphs with unit length edges therefore implies Theorem \ref{thm:kep}.

    The proof for Theorem \ref{thm:kep} on graphs with unit edge lengths is a slight modification from the $\tO(k^{1/2 + \ep})$-competitive online algorithm on graphs with unit edge lengths and costs in \cite{grigorescu2021online}. The outline is as follows.

    \begin{enumerate}
        \item We consider \opwspl on a graph $G'$ consisting of disjoint layered graphs constructed from the original graph $G$ by losing a constant factor. The layers allow us to capture the distance constraints. This allows us to further reduce the problem to \oslc on $G'$.
        \item We further consider \oslc on an undirected forest $H$ constructed from $G'$ with a loss of an $O(k^{\ep})$ factor via the height reduction technique introduced in \cite{chekuri2011set,helvig2001improved}. 
    \end{enumerate}

    We start with the definition for \slc introduced in \cite{grigorescu2021online}.

    \begin{restatable}{definition}{defslc}\label{def:slc}
    In the \emph{Steiner label cover} problem, we are given a (directed or undirected) graph $G=(V,E)$, non-negative edge costs $w: E \to \R_{\geq 0}$, and a collection of $k$ disjoint vertex subset pairs $(S_i, T_i)$ for $i \in [k]$ where $S_i, T_i \subseteq V$ and $S_i \cap T_i = \emptyset$. Each pair is associated with a relation (set of permissible pairs) $R_i \subseteq S_i \times T_i$. The goal is to find a subgraph $F=(V,E')$ of $G$, such that 1) for each $i \in [k]$, there exists $(s,t) \in R_i$ such that there is an $s \leadsto t$ path in $F$, and 2) the cost $\sum_{e \in E'}w(e)$ is minimized.
\end{restatable}

In \oslc, $(S_i, T_i)$ and $R_i$ arrive online, and the goal is to irrevocably select edges to satisfy the first requirement and also approximately minimize the cost.

Given an instance of \opwspl on a graph with unit edge lengths, we construct the graph $G'$ by following the construction of $G_r$ in the proof of Lemma \ref{le:poly_junction_tree_weighted_distance} in Appendix \ref{app:sec:wjt} and taking the union of $G_r$ for all $r \in V$. That is, we \emph{simultaneously} consider all the vertices $r \in V$ as a candidate root for the distance-preserving weighted junction tree solution while constructing the online solution.

Recall that $G_r=(V_r,E_r)$ is constructed as follows (see Figure \ref{fig:lg} for an example):
\begin{equation*}
    V_r = ((V\setminus r) \times \{-n+1,\ldots,-2,-1,1,2,\ldots,n-1\}) \cup \{(r,0)\}
\end{equation*}
and
\begin{equation*}
    E_r = \{((u,i),(v,i+1)) \mid (u,i),(v,i+1) \in V_r,(u,v)\in E\}.
\end{equation*}
For each edge $e = ((u,i),(v,i+1)) \in E_r$, let the edge weight be $w(e) = c((u,v))$. This is the only change we make from \cite{grigorescu2021online} and it has very little impact on the overall proof. For each vertex $(u,-j)$ (respectively $(u,j)$) in $V_r$ where $j \in [n-1] \cup \{0\}$, add a vertex $(u^-,-j)$ (respectively $(u^+,j)$), and create an edge $(u^-,-j) \to (u,-j)$ (respectively $(u,j) \to (u^+,j)$) with weight 0. This concludes the construction of $G_r$.

Let $G'$ be the disjoint union of $G_r$ for $r \in V$. Given an \opwspl instance on $G$, we can construct an \oslc instance on $G'$. 
For explicitness, we denote the vertex $(u,j)$ in $V_r$ by $(u,j)_r$. For each $(s_i, t_i)$, let $S_i = \{(s^-_i, -j)_r \mid j > 0, r \in V \setminus \{s_i\}\} \cup (s^-_i,0)_{s_i}$, $T_i = \{(t^+_i, j)_r \mid j > 0, r \in V \setminus \{t_i\}\} \cup (t^+_i,0)_{t_i}$, and $R_i = \{((s^-_i, -j_s)_r,(t^+_i, j_t)_r) \mid r \in V, j_s+j_t \leq Dist(s_i,t_i)\}$. Intuitively, a copy of $s_i$ and a copy of $t_i$ belongs to the relation $R_i$ if 1) they are connected by a distance-preserving weighted junction tree with the same root $r$, and 2) the distance between them is at most $Dist(s_i,t_i)$ in this distance-preserving weighted junction tree.

We can recover a solution for \opwspl on $G$ with at most the cost of the corresponding \oslc solution in $G'$. Suppose we have the \oslc solution in $G'$, for each selected edge $(u,j) \to (v,j+1)$ in $G_r$, we select $u \to v$ for the pairwise spanner solution in $G$. The optimum of the Steiner label cover solution in $G'$ is at most $2\opt_{junc}$. In each distance-preserving weighted, each edge is used at most once in the in-arborescence and at most once in the out-arborescence. The following claim is from \cite{grigorescu2021online}.

\begin{claim} \label{cl:layer-ps-opt}
The optimal value of the \slc problem in $G'$ is at most $2\opt_{junc}$.
\end{claim}

The remaining is to show how height reduction is used to find a solution for \oslc on $G'$. We recall Lemma \ref{le:height_reduction}.

\lemhr*

We construct a tree $H_r$ from $G_r$ by the height reduction technique. We recall that $G_r$ is a layered graph with a center vertex $(r,0)$. Let $G^-_r$ denote the subgraph of $G_r$ induced by the vertex set
\[\tuple{\bigcup_{v \in V \setminus \{r\}}\{v, v^-\} \times \{-j \mid j \in [n-1]\}} \cup \{(r,0), (r^-,0)\},\]
and similarly, let $G^+_r$ denote the subgraph of $G_r$ induced by the vertex set
\[\tuple{\bigcup_{v \in V \setminus \{r\}}\{v, v^+\} \times [n-1]} \cup \{(r,0), (r^+,0)\}.\]
By Lemma~\ref{le:height_reduction}, with $(r,0)$ being the root, we can construct a tree $T_r^-$ rooted at $r^-$ for $G_r^-$ which approximately preserves the cost of any in-arborescence rooted at $(r,0)$ in $G_r^-$ by a subtree in $T_r^-$, and similarly a tree $T_r^+$ rooted at $r^+$ for $G_r^+$ which approximately preserves the cost of any out-arborescence rooted at $(r,0)$ in $G_r^+$ by a subtree in $T_r^+$. We further add a super root $r'$ and edges $\{r',r^-\}$ and $\{r',r^+\}$ both with weight 0. This concludes the construction of the weighted tree $H_r$ (see Figure~\ref{fig:hr} for an illustration).

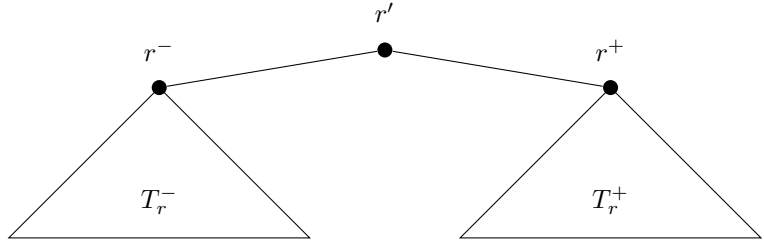
\begin{figure}[h]
\centering
\begin{tikzpicture}
    \draw (-2,0) -- (2,0) -- (0,2) --cycle;
    \node[draw=none] at (0,0.5) {$T_r^-$};
    \draw (4,0) -- (8,0) -- (6,2) --cycle;
    \node[draw=none] at (6,0.5) {$T_r^+$};
    \node at (0,2.5) {$r^-$};
    \node at (6,2.5) {$r^+$};
    \node at (3,3) {$r'$};
    \node[fill,circle, inner sep=0pt, minimum size=0.2cm] (t1) at (0,2) {};
    \node[fill,circle, inner sep=0pt, minimum size=0.2cm] (t2) at (6,2) {};
    \node[fill,circle, inner sep=0pt, minimum size=0.2cm] (t3) at (3,2.5) {};
    \path[-]
        (t1) edge (t3)
        (t2) edge (t3);
\end{tikzpicture}
\caption{Construction of $H_r$. The figure is from \cite{grigorescu2021online}.}
\label{fig:hr}
\end{figure}

Let $H$ be the disjoint union of $H_r$ for $r \in V$. We show that we can achieve an $\tO(k^\ep)$-approximation for the Steiner label cover problem on graph $H$ in an online manner. We set $\sigma=\lceil 1/\ep \rceil$ and apply Lemma~\ref{le:height_reduction} to obtain the weight $\hat{w}:E(H) \to \R_{\geq 0}$ and the mapping $\Psi_r: V(H_r) \to V(G_r)$ for all $r \in V$. Let $\Psi$ be the union of the mappings for all $r \in V$. For each pair $(s_i,t_i)$ in the original graph $G$, we recall that we focus on the vertex subset pair $(S_i,T_i)$ of $G'$ and its relation $R_i$ which captures the distance requirement.

To establish the correspondence between the Steiner label cover instances, we clarify the mapping between the leaves of the arborescences in $G'$ and $H$. Given a vertex $(s^-_i,-j)_r$ in $V(G')$ with a non-negative $j$, let $\Psi^{-1}((s^-_i,-j)_r)$ denote the set of leaves in $T^-_r$ that maps to $(s^-_i,-j)_r$ by $\Psi$. Similarly, given a vertex $(t^+_i,j)_r$ in $V(G')$ with a non-negative $j$, let $\Psi^{-1}((t^+_i,j)_r)$ denote the set of leaves in $T^+_r$ that maps to $(t^+_i,j)_r$ by $\Psi$. The mapping $\Psi^{-1}$ naturally defines the terminal sets of interest $\hat{S_i}:=\Psi^{-1}(S_i)=\{\hat{s} \in V(H) \mid \Psi(\hat{s}) \in S_i\}$ and $\hat{T_i}:=\Psi^{-1}(T_i)=\{\hat{t} \in V(H) \mid \Psi(\hat{t}) \in T_i\}$. The relation is also naturally defined: $\hat{R}_i:= \{(\hat{s},\hat{t}) \in \hat{S}_i \times \hat{T}_i \mid (\Psi(\hat{s}),\Psi(\hat{t})) \in R_i\}$. We note that for $\hat{s} \in \hat{S}_i$ and $\hat{t} \in \hat{T}_i$, $(\hat{s},\hat{t})$ belongs to $\hat{R}_i$ only when $\hat{s}$ and $\hat{t}$ belong to the same tree $H_r$ (see Figure \ref{fig:sc} for an illustration).

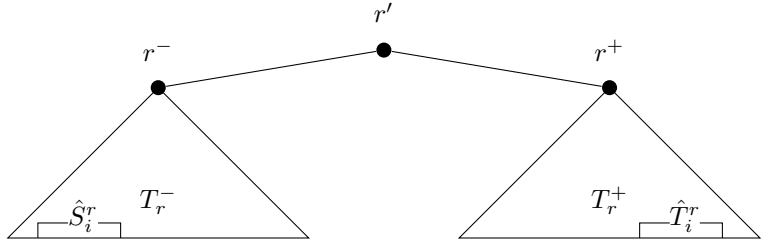
\begin{figure}[h]
\centering
\begin{tikzpicture}
    \draw (-2,0) -- (2,0) -- (0,2) --cycle;
    \node[draw=none] at (0,0.5) {$T_r^-$};
    \draw (4,0) -- (8,0) -- (6,2) --cycle;
    \node[draw=none] at (6,0.5) {$T_r^+$};
    \node at (0,2.5) {$r^-$};
    \node at (6,2.5) {$r^+$};
    \node at (3,3) {$r'$};
    \node[fill,circle, inner sep=0pt, minimum size=0.2cm] (t1) at (0,2) {};
    \node[fill,circle, inner sep=0pt, minimum size=0.2cm] (t2) at (6,2) {};
    \node[fill,circle, inner sep=0pt, minimum size=0.2cm] (t3) at (3,2.5) {};
    \path[-]
        (t1) edge (t3)
        (t2) edge (t3);
    \node at (-1,0.3) {$\hat{S}^r_i$};
    \node at (7,0.3) {$\hat{T}^r_i$};
    \draw (-1.6,0) -- (-1.6,0.2) -- (-1.2,0.2);
    \draw (-0.8,0.2) -- (-0.5,0.2) -- (-0.5,0);
    \draw (8-1.6,0) -- (8-1.6,0.2) -- (8-1.2,0.2);
    \draw (8-0.8,0.2) -- (8-0.5,0.2) -- (8-0.5,0);
\end{tikzpicture}
\caption{From \opwspl to \oslc. In $G$, consider connecting $s_i$ and $t_i$ by using a distance-preserving weighted junction tree rooted at $r$. In $H_r$, we consider terminal sets $\hat{S}^r_i := \hat{S}_i \cap H_r$ and $\hat{T}^r_i := \hat{T}_i \cap H_r$. A distance-preserving weighted junction tree in $G$ rooted at $r$ corresponds to a subtree with leaves that belongs to $\hat{S}^r_i \cup \hat{T}^r_i$. The set $\hat{R}^r_i := \hat{R}_i \cap (\hat{S}^r_i \times \hat{T}^r_i)$ defines the admissible pairs, which maintains the distance requirement $Dist(s_i,t_i)$ in the layered graph $G_r$. The figure is from \cite{grigorescu2021online}.}
\label{fig:sc}
\end{figure}

We use the following theorem from \cite{grigorescu2021online}.

\begin{theorem} \label{thm:oslc}
    For the \oslc problem on $H, \hat{S}_i$, $\hat{T}_i$, $\hat{R}_i$, there is a randomized polynomial-time algorithm with competitive ratio $\polylog(n)$.
\end{theorem}

Combining Theorem \ref{thm:oslc}, Lemma \ref{le:height_reduction}, and Claim \ref{cl:layer-ps-opt}, we can recover a solution of cost at most $\tO(k^{\ep}) \opt_{junc}$ for \opwspl on $G$ from a solution of \oslc on $H$.
\end{proof}

\subsection{Missing proof for Lemma~\ref{lem:sqrt-k-den}} \label{pf:lem:sqrt-k-den}

\lemsqrtkden*

\begin{proof}
Let $G^*$ (a subgraph of $G$) be the optimal pairwise weighted spanner solution with cost $\opt$. The proof proceeds by considering the following two cases: 1) there exists a vertex $r \in V$ that belongs to at least $\sqrt{k}$ $s_i \leadsto t_i$ paths of distance at most $Dist(s_i,t_i)$ in $G^*$ for distinct $i$, and 2) there is no such vertex $r \in V$.

For the first case, we consider the union of the $s_i \leadsto t_i$ paths in $G^*$, each of distance at most $Dist(s_i,t_i)$, that passes through $r$. This subgraph in $G^*$ contains an in-arborescence and an out-arborescence both rooted at $r$, whose union forms a distance-preserving weighted junction tree. This distance-preserving weighted junction tree has cost at most $\opt$ and connects at least $\sqrt{k}$ terminal pairs, so its density is at most $\opt / \sqrt{k}$.

For the second case, each vertex $r \in V$ appears in at most $\sqrt{k}$ $s_i \leadsto t_i$ paths in $G^*$. More specifically, each edge $e \in E$ also appears in at most $\sqrt{k}$ $s_i \leadsto t_i$ paths in $G'$. By creating $\sqrt{k}$ copies of each edge, all terminal pairs can be connected by edge-disjoint paths. Since the overall duplicate cost is at most $\sqrt{k} \cdot \opt$, at least one of these paths has cost at most $\sqrt{k} \cdot \opt / k $. This path constitutes a distance-preserving weighted junction tree whose density is at most $\opt / \sqrt{k}$.
\end{proof}